\documentclass[a4paper,UKenglish,cleveref, autoref, thm-restate]{lipics-v2021}

\pdfoutput=1 
\hideLIPIcs  

\usepackage{todonotes}

\title{Shortest Beer Path Queries in Interval Graphs} 

\author{Rathish Das}{Department of Computer Science, University of Liverpool, UK\and \url{https://cs.uwaterloo.ca/~r27das/}}{rathish.das@liverpool.ac.uk}{}{}

\author{Meng He}{Faculty of Computer Science, Dalhousie University, Halifax, Canada\and \url{https://web.cs.dal.ca/~mhe/}}{mhe@cs.dal.ca}{}{}

\author{Eitan Kondratovsky}{Cheriton School of Computer Science, University of Waterloo, Canada \and \url{https://uwaterloo.ca/scholar/e2kondra/home}}{eitan.kondratovsky@uwaterloo.ca}{}{}

\author{J. Ian Munro}{Cheriton School of Computer Science, University of Waterloo, Canada \and \url{https://cs.uwaterloo.ca/~imunro/}}{imunro@uwaterloo.ca}{}{}

\author{Anurag Murty Naredla}{Cheriton School of Computer Science, University of Waterloo, Canada \and \url{}}{amnaredla@uwaterloo.ca}{}{}

\author{Kaiyu Wu}{Cheriton School of Computer Science, University of Waterloo, Canada}{k29wu@uwaterloo.ca}{https://orcid.org/0000-0001-7562-1336}{}

\authorrunning{Rathish Das, Meng He, Eitan Kondratovsky, J. Ian Munro, Anurag Murty, Kaiyu Wu}

\Copyright{Rathish Das, Meng He, Eitan Kondratovsky, Ian Munro, Anurag Murty, Kaiyu Wu}



\category{} 

\relatedversion{} 

\nolinenumbers 

\begin{document}
\newcommand{\qbsp}{\texttt{beer\_shortest\_path}}
\newcommand{\qbdist}{\texttt{beer\_dist}}
\newcommand{\qsp}{\texttt{shortest\_path}}
\newcommand{\qdist}{\texttt{dist}}
\newcommand{\qadj}{\texttt{adjacent}}
\newcommand{\qdeg}{\texttt{degree}}
\newcommand{\qnb}{\texttt{neighbourhood}}
\newcommand{\qrank}{\texttt{rank}}
\newcommand{\qselect}{\texttt{select}}
\newcommand{\qaccess}{\texttt{access}}
\newcommand{\qdepth}{\texttt{depth}}
\newcommand{\qpost}{\texttt{post}}
\newcommand{\qlast}{\texttt{last}}
\newcommand{\qpred}{\texttt{pred}}

\maketitle

\begin{abstract}
Our interest is in paths between pairs of vertices that go through at least one of a subset of the vertices known as beer vertices. Such a path is called a beer path, and the beer distance between two vertices is the length of the shortest beer path.

We show that we can represent unweighted interval graphs using $2n \log n + O(n) + O(|B|\log n)$ bits where $|B|$ is the number of beer vertices. This data structure answers beer distance queries in $O(\log^\varepsilon n)$ time for any constant $\varepsilon > 0$ and shortest beer path queries in $O(\log^\varepsilon n + d)$ time, where $d$ is the beer distance between the two nodes. We also show that proper interval graphs may be represented using $3n + o(n)$ bits to support beer distance queries in $O(f(n)\log n)$ time for any $f(n) \in \omega(1)$ and shortest beer path queries in $O(d)$ time. All of these results also have time-space trade-offs.

Lastly we show that the information theoretic lower bound for beer proper interval graphs is very close to the space of our structure, namely $\log(4+2\sqrt{3})n - o(n)$ (or about $ 2.9 n$) bits.
\end{abstract}
\section{Introduction}
The concept of a beer path was recently introduced by Bacic et al. \cite{DBLP:conf/isaac/Bacic0S21}. The premise is simple, suppose you wish to visit a friend, and wish to pick up some beer along the way because you don't want to show up empty handed, what is the fastest way to do so?
More formally, for a graph, we specify a set of vertices, which will act as beer stores. A beer path is one which passes through at least one of these designated vertices. We will say a beer graph is one where we have designated a subset of the vertices to be beer stores.
Though this premise may be somewhat silly, it can have many applications. For example, suppose you are going on a road trip and to be efficient, want to drop something off at a post office on the way. Or perhaps on your trip, you realize that you currently don't have enough gas, so you must visit a gas station somewhere along the way. Another hypothetical situation would be if a package needs to be transported, but due to regulations, one of the stops must be equipped for an inspection.

It is easily seen that a shortest beer path may not be simple, but will always consist of two shortest paths: from the beer store to the source, and to the destination.

In this paper, we study the shortest beer path problem on unweighted interval graphs: intersection graphs of intervals on the real line. Interval graphs are a well-known class of graphs and have applications in operations research \cite{DBLP:journals/jacm/Bar-NoyBFNS01} and bioinformatics \cite{DBLP:journals/bioinformatics/ZhangSFCWKB94}. For a more indepth treatment of interval graphs and their applications, see the book of Golumbic \cite{golumbic2004algorithmic}.

\textbf{Related Work:} 
Bacic et al. \cite{DBLP:conf/isaac/Bacic0S21} studied the problem on weighted outerplanar graphs. They showed that on an outerplanar graph of $n$ vertices, a data structure of size $O(m)$ words for any $m \ge n$ can be constructed in $O(m)$ time to support shortest beer path and beer distance - the length of the shortest beer path in $O(\alpha(m,n))$ time, where $\alpha$ is the inverse Ackermann function.

On the more general problem of graph data structures, Acan et al. \cite{DBLP:journals/algorithmica/AcanCJS21} showed that interval graphs may be represented succinctly using $n\log n +O(n)$ bits of space to answer basic navigational queries: $\qadj, \qdeg, \qnb$ and $\qsp$ in optimal time: $O(1)$ or $O(1)$ for each vertex in the output. 
Building on Acan et al.'s work, He et al. \cite{DBLP:conf/isaac/0001MNWW20} added the $\qdist$ query for interval graphs. Their data structure has the same space $n\log n + O(n)\footnote{We will use $\log$ to denote $\log_2$}$ bits, the same run time of the old operations but also supports $\qdist$ in $O(1)$ time.

\subsection{Our Results and Paper Layout}
We give data structures for beer interval graphs and beer proper interval graphs that have time-space trade offs. 
An interval graph is a graph where we may assign an interval on the real line to each vertex - $v \mapsto [ l_v,r_v ]$ such that two vertices $u,v$ are adjacent exactly when the corresponding intervals intersect \cite{hajos1957art,lekkeikerker1962representation}.
A proper interval graph is an interval graph where the intervals must be chosen so that no two intervals nest. Furthermore, we prove a lower bound result on the space required for beer proper interval graphs.

The main obstacle in constructing the data structures for the beer path queries is that the set of paths between two vertices (which are normally the feasible solutions to the shortest path problem) are arbitrarily filtered by the beer nodes into a smaller set of feasible solutions to the beer shortest path problem - by whether a beer node exists on the path or not. In the case that a beer node exists on one of the shortest paths, then it is clearly optimal and we must be able to detect this. Thus we must be able conduct this filtering process as well, and we achieve this by using orthogonal range search on the previously established data structures which looks at all paths. 


In section \ref{s:ds-proper} we study the beer distance problem in proper interval graphs. We first outline the steps that we need to implement in our data structures. Then in subsection \ref{ss:ds-proper-1} we give a data structure occupying $3n + o(n) + O(|B|\log n)$ bits of space, supporting all regular operations in the same complexity as the previous works of Acan et al. and He et al., and the queries related to beer distance:
\begin{itemize}
    \item $\qbdist$ in $O(\log^\varepsilon n)$ time, for any constant $\varepsilon > 0$
    \item $\qbsp$ in $O(1)$ time per vertex on the path.
\end{itemize}

We may also utilize a trade-off provided by our auxiliary data structures, which decrease the query times by substituting $\log^\varepsilon n$ with $\log\log n$ at the cost of increasing one of the space cost terms from $O(|B|\log n)$ to $O(|B|\log n\log\log n)$ bits. We note that in this data structure, the space is dependent on $|B|$ the number of beer vertices, and is therefore undesirable if $|B|$ is large. In the case that there are many beer vertices, the above data structure can use $O(n\log n)$ bits of space. 

In subsection \ref{s:improved-ds}, we use the tree structure of the distance queries to eliminate the dependence on $|B|$ at the cost of slightly increasing the run time. For beer proper interval graphs, we have a data structure using $3n + o(n)$ bits of space, which supports all the regular queries in their original optimal complexities, and the beer queries:
\begin{itemize}
    \item $\qbdist$ in $O(f(n)\log n)$ time for any function $f(n) \in \omega(1)$. Different $f(n)$ will impact the lower order term $o(n)$ in the space complexity.
    \item $\qbsp$ in $O(1)$ per vertex on the path.
\end{itemize}
In section \ref{s:ds-interval} we study the beer distance problem in interval graphs. For this, we give a data structure using $2n\log n + O(n) + O(|B|\log n)$ bits, where $|B|$ is the number of beer vertices in the graph, and supports all regular operations in the same time complexity as above, and 
\begin{itemize}
    \item $\qbdist$ in $O(\log^\varepsilon n)$ time, for any constant $\varepsilon > 0$.
    \item $\qbsp$ in $O(\log^\varepsilon n + d)$ time, where $d$ is the beer distance between the two vertices. 
\end{itemize}

Again we may utilize the same trade off to replace the $\log^\varepsilon n$ term by $\log\log n$ at the cost of increase the space term $|B|\log n$ to $|B|\log n\log\log n$.

Finally in section \ref{s:lower-bound} we count the number of non-isomorphic beer proper interval graphs and use this to give an information theoretic lower bound on the space required for any data structure for beer proper interval graphs that can support $\qadj$ and $\qbdist$. It may seem natural that to store which vertices are beer nodes will require an additional $n$ bits but we show that the lower bound is actually asymptotically $\log(4+2\sqrt{3})n \approx 2.9n$ bits. The main insight into seeing why an additional $n$ bits is not required is that for a clique, it suffices to only store $\log n$ bits for the count of how many beer nodes. For general interval graphs, the space required for the beer nodes is at most $n$ and is a lower order term to the lower bound of $n\log n$ bits. Thus there is nothing study in this case.

Finally in Appendix \ref{s:bounding}, we give an obsolete way to bound the number of beer proper interval graphs, which is superseded by the more exact analysis in Section \ref{s:lower-bound}. However the techniqued used may be useful in it own right and may also be useful for the construction of a data structure for beer proper interval graphs that uses fewer than $3n$ bits of space.

\section{Preliminaries}
\label{s:prelim}
In this paper, we will use the standard graph theoretic notation. We will use $G = (V,E)$ to denote a graph with vertex set $V$ and edge set $E$. We will use $n=|V|$ and $m=|E|$ to denote the number of vertices and edges. All of our graphs will be unweighted.

As we will be discussing both trees and graphs in general, we will use vertices to denote the vertices of a graph which may or may not be a tree, and nodes to denote the vertices of a tree.
In the paper, we assume the word-RAM model with $\Theta(\log n)$-size words. We use $\log(
\cdot)$ to denote $\log_2(
\cdot)$.

In a beer graph, we take any underlying graph $G$ together with a set $B \subseteq V$ of \emph{beer vertices}. This allows us to define the following queries:

\begin{itemize}
\item \qbsp$(u,v)$: return a shortest path between the vertices $u$ and $v$ such that at least one of the beer vertices appears on the path.

\item \qbdist$(u,v)$: return the length of the shortest path between vertices $u$ and $v$ such that at least one of the beer vertices appears on the path.
\end{itemize}

These are the restricted queries to the ordinary $\qsp$ and $\qdist$ queries, which do not have the constraint that it must pass through a beer vertex.

For example, if $B = V$, then the two queries reduces to ordinary shortest path or distance in the graph. On the other extreme, if $B = \{b\}$ is a singleton, then the query reduces to two ordinary shortest path or distance queries in the graph.

\subsection{Interval Graphs}
An interval graph $G$ is a graph where we may assign an interval on the real line to each vertex - $v \mapsto [ l_v,r_v ]$ such that two vertices $u,v$ are adjacent exactly when the corresponding intervals intersect \cite{hajos1957art}. In particular, we may sort the endpoints so that the values are integers between $1$ and $2n$. We sort the vertices based on their left endpoints, so that when we refer to vertex $v$, we are referring to a number (the rank of the vertex in the sorted order), and thus, a statement such as $u<v$ makes sense.

Acan et al. \cite{DBLP:journals/algorithmica/AcanCJS21} showed that interval graphs can be represented using $n\log n + O(n)$ bits to support \qadj, \qdeg, \qnb, {\qsp} queries in optimal time. $\qadj$ answers whether two given vertices are adjacent, {\qdeg} answers the degree of the given vertex, {\qnb} gives a list of the neighbours of the given vertex, and {\qsp} gives a shortest path between the two given vertices.

He et al. \cite{DBLP:conf/isaac/0001MNWW20} showed that we can also answer the $\qdist$ query in optimal time. Therefore we have the following theorem on interval graphs:

\begin{lemma} An interval graph $G$ on $n$ vertices can be represented succinctly using $n\log n + O(n)$ bits to support $\qadj, \qdeg$ and $\qdist$ queries in $O(1)$ time, $\qnb$ in $O(1)$ time per neighbour and $\qsp$ in $O(1)$ time per vertex on the path.
\end{lemma}

An interval graph $G$ is \emph{proper} (or a \emph{proper interval graph}) if we can choose the intervals corresponding to vertices such that no two intervals are nested.

Acan et al. \cite{DBLP:journals/algorithmica/AcanCJS21} also showed that proper interval graphs can be represented using $2n + o(n)$ bits to support $\qadj, \qdeg, \qnb, \qsp$ queries in optimal time. He et al. \cite{DBLP:conf/isaac/0001MNWW20} showed that we may also support the $\qdist$ query in optimal time. Thus we have the following theorem on proper interval graphs:

\begin{lemma}
	A proper interval graph $G$ on $n$ vertices can be represented succinctly using $2n + o(n)$ bits to support $\qadj, \qdeg$ and $\qdist$ queries in $O(1)$ time, $\qnb$ in $O(1)$ time and $\qsp$ in $O(1)$ time per vertex on the path.
\end{lemma}

\subsection{Dyck Paths}
For our lower bound, we will be discussing Dyck paths. A Dyck path of length $2n$ is a path from $(0,0)$ to $(2n,0)$ using $2n$ steps, $n$ of which are $(1,1)$ steps which are referred to as up-steps, and $n$ of which are $(1,-1)$ steps and are referred to as down-steps. Such a path must also satisfy the condition that it never reaches below the $x$-axis. It is well known that the number of Dyck paths of length $2n$ is $C_n=\frac{1}{n+1}\binom{2n}{n}$ the $n$-th Catalan number.

A Dyck path that never touches the $x$-axis except at the start and end, is referred to as an irreducible Dyck path. By removing the up-step at the beginning and the down-step at the end, the remainder of the path is simply a Dyck path of length $2(n-1)$. Thus the number of irreducible Dyck paths of length $2n$ is simply $C_{n-1}$.

For any Dyck path, we may associate an up-step with an open parenthesis $($ and a down-step with a close parenthesis $)$. The sequence we obtain from a Dyck path is a balanced parenthesis sequence (and vice versa) as the Dyck path condition is exactly the condition that the excess in the balanced parenthesis sequence is never negative. This well known bijection allows us to associate an forest to each Dyck path, using the well known bijection for forests and balanced parentheses (via a depth-first traversal). In particular, if the Dyck path were irreducible, then the forest is just a single tree.

\subsection{Succinct Data Structures}
The information theoretic lower bound to represent a family of objects with $N$ elements is $\lceil\log N\rceil$ bits. Any fewer bits and we do not have enough bit strings to assign a unique one to each object, and thus cannot distinguish between them. A succinct data structure aims to use $\log N + o(\log N)$ bits to represent these objects while supporting the relevant queries.

A bit vector is a length $n$ array of bits, that supports the queries $\qrank(i)$: given an index, return the number of 1s up to index $i$, $\qselect(j)$: given a number $j$, return the index of the $j$th one in the array, and $\qaccess(i)$: return the bit at index $i$.

\begin{lemma}[Munro et al. \cite{DBLP:journals/jal/MunroRR01}]
	A bit vector of length $n$ can be succinctly represented using $n + o(n)$ bits to support $\qrank,\qselect$ and $\qaccess$ in $O(1)$ time.
\end{lemma}

We will also require the compressed form of bit-vectors, where if the number of 1s is small, we are able to get away with using less space.

\begin{lemma}[Patrascu \cite{DBLP:conf/focs/Patrascu08}]
	A bit vector of length $n$ with $m$ 1s can be represented using $\log \binom{n}{m} + O(\frac{n}{\log^c n} + m) \le m\log \frac{n}{m} + O(\frac{n}{\log^c n} + m)$ for any constant $c$. The data structure supports $\qrank,\qselect,\qaccess$ in $O(1)$ time.
\end{lemma}

We will be working with trees and thus will be discussing and using various tree operations. These will mainly be conversions between the nodes' numbers in the different traversals: pre-order, post-order, level-order. These operations can be done in $O(1)$.

\begin{lemma}[He et al. \cite{DBLP:conf/isaac/0001MNWW20}]
    An ordinal tree on $n$ nodes can be represented succinctly using $2n + o(n)$ bits and can support a variety of operations in $O(1)$ time. For the full list see table 1 in their paper.
\end{lemma}

\subsection{Orthogonal Range Queries}
In these data structures, we store $n$ $d$-dimensional points. The queries we wish to answer are: given a $d$-dimensional axis aligned rectangle $[p_1,p_2]\times[p_3,p_4]\ldots[p_{2d-1},p_{2d}]$ (we use the closed intervals here in the definition, but when our points are integers it is easy to see how to support open or semi-open intervals as well), \emph{emptiness}: does it contain a point? \emph{count}: how many points does it contain? \emph{reporting}: return each point. We say that the rectangle is $k$-sided if there is at most $k$ coordinates that are finite (in general the coordinates $p_i$ may be $\pm \infty$). When we do not mention how many sides, it is assumed that it is the maximum possible: $2d$.

When $d = 2$, Chan et al. \cite{DBLP:conf/compgeom/ChanLP11} showed that we solve the emptiness problem:
\begin{lemma}
	We can solve the 2d range emptiness queries using $O(n\log n\log\log n)$ bits of space and $O(\log \log n)$ query time or $O(n\log n)$ bits of space and $O(\log^\epsilon n)$ query time for any constant $\epsilon > 0$.
	
\end{lemma} 

When $d=3$, Nekrich \cite{DBLP:conf/soda/Nekrich21} showed that we may solve the emptiness and reporting problems:\begin{lemma}
	\label{t:3d-report}
	For $n$ 3D points and constant $\epsilon > 0$, we may support 5-sided orthogonal reporting queries using $O(n\log n)$ bits of space and $O(k\log^\epsilon n)$ time or $O(n\log n\log\log n)$ bits of space and $O(k\log \log n)$ time, where $k$ is the size of the output.
	
	By setting $k =0$, we may achieve the complexities to emptiness queries as well.
\end{lemma}

\subsection{Predecessor Queries}
Given a set of numbers $S$ in a universe $U$, we wish to answer the following query: $\qpred(i)$, given an element $i$ of $U$, return the largest element $j$ of $S$ that is smaller than $i$.

Though there have been a lot of work on this problem, we will only need one of the basic results by Willard \cite{DBLP:journals/ipl/Willard83} which gives a $O(N)$ word space solution to the problem with $O(\log\log U)$ query time.

\begin{lemma}
    There is a data structure for the predecessor problem that uses $O(N)$ words of space and query time $O(\log\log U)$.
\end{lemma}

\section{Beer Paths in Proper Interval Graphs}
\label{s:ds-proper}
In this section we investigate beer paths in proper interval graphs.
We will be using the data structure of He et al. \cite{DBLP:conf/isaac/0001MNWW20} as it supports the $\qdist$ operation, which we will modify to account for the beer vertices.
We begin with an example:
\begin{example}
\label{e:proper-interval}
Consider the graph with the interval representation given by the bit string: $000001000101001110011011011111$. This gives the vertex 1 a left endpoint at coordinate 1 and right index at coordinate 6 (the first 0 and first 1 in the sequence respectively). A graphical representation of the graph and the corresponding distance tree is given.
\begin{center}
    
\includegraphics[scale=0.4]{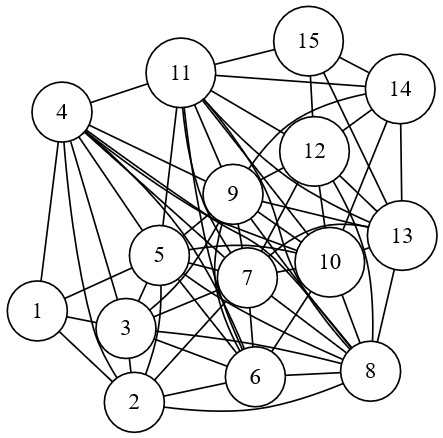}
\includegraphics[scale=0.25]{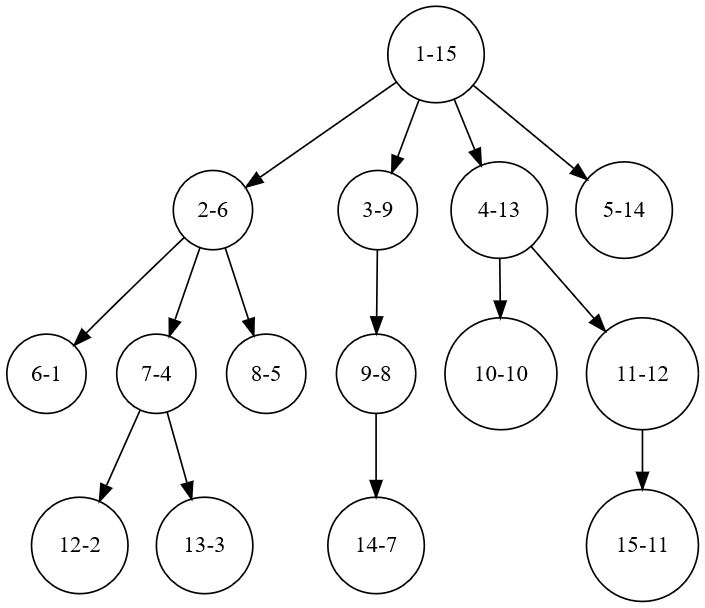}
\end{center}

The first number in the node of the distance tree is the node's level-order number and the second its post-order traversal number.

Consider the shortest path between nodes 13 and 3. The shortest path algorithm given by He et al. \cite{DBLP:conf/isaac/0001MNWW20} will give the path $13 \rightarrow  7 \rightarrow 3$.
The problem would be easy if one of these nodes were a beer node, say node 7 then there would be far less to do. However consider the case that the only beer node were node 6, then a $\qbsp$ would be $13 \rightarrow 7 \rightarrow 6 \rightarrow 3$.
On the other hand, if there were also a beer node at node 8, then we can take the path $13 \rightarrow 8 \rightarrow 3$.
\end{example}

\subsection{Calculating Beer Distance}

In the data structure of He et. al \cite{DBLP:conf/isaac/0001MNWW20}, we represent the proper interval graph $G$ using a distance tree $T$. There is a bijection between the vertices of the graph $G$ and the nodes of the tree - vertex $v$ is mapped to the $v$th node in the tree in level-order. Thus by $v$ we will simultaneously refer to the vertex and the node in the distance tree. As the conversion between level-order ranks, pre-order ranks and post-order ranks in a tree are all $O(1)$ (and typically, nodes in a tree are referred by their pre-order ranks), we will implicitly convert between them as the situation requires. All the queries are reduced to tree operations and can be done in $O(1)$ time. 

Since we will be building upon the $\qdist$ and $\qsp$ queries, we will explain it in detail. The distance tree's parent child relationship is the following: for a node $v$, the parent of $v$ is the smallest (indexed) node that is adjacent to $v$.

Let $u < v$, be the nodes involved in the $\qdist/\qsp$ query, and suppose that $\qdepth(u) = k_1 \le k_2 = \qdepth(v)$. We repeatedly take the parent of $v$ to form the following chain: $v_{k_1},v_{k_1 +1},\ldots,v_{k_2} = v$, where $v_j$ is at depth $j$. If $v_{k_1} < u$ then a shortest path is $u,v_{k_1 +1},\ldots,v_{k_2} = v$. If $v_{k_1} = u$, then the shortest path is $u = v_{k_1},v_{k_1 +1},\ldots,v_{k_2} = v$. Finally if $v_{k_1} > u$, then a shortest path is $u, v_{k_1}, v_{k_1 +1},\ldots,v_{k_2} = v$.
To compute the length without getting every node, we simply subtract the depths of $u$ and $v$, use level-ancestor to find $v_{k_1}$ and find which case we fall into to adjust the distance.

We note that this is only one of many possible shortest paths. To accommodate the beer vertices, we will investigate what all the possible shortest paths might look like. To this end, we will consider the following question: let $u < w < v$, is $w$ on a shortest path between $u$ and $v$? Equivalently, is $\qdist(u,v) = \qdist(u,w) + \qdist(w,v)$? In this case we say that $w$ \emph{preserves the distance}.
Let $\qpost(v)$ denote the post-order rank of a node in the tree. It is clear that if $\qpost(u) < \qpost(v)$, then $u$ is to the left of $v_{k_1}$, so that $u < v_{k_1}$, and similarly for the reversed inequality.
Finally, we note that for nodes on the same level of the tree, their post-order numbers are sorted. That is if $u < v$ are on the same level, then $\qpost(u) < \qpost(v)$ (and vice versa).
\begin{figure}[h]
    \centering
    \includegraphics[width=\textwidth]{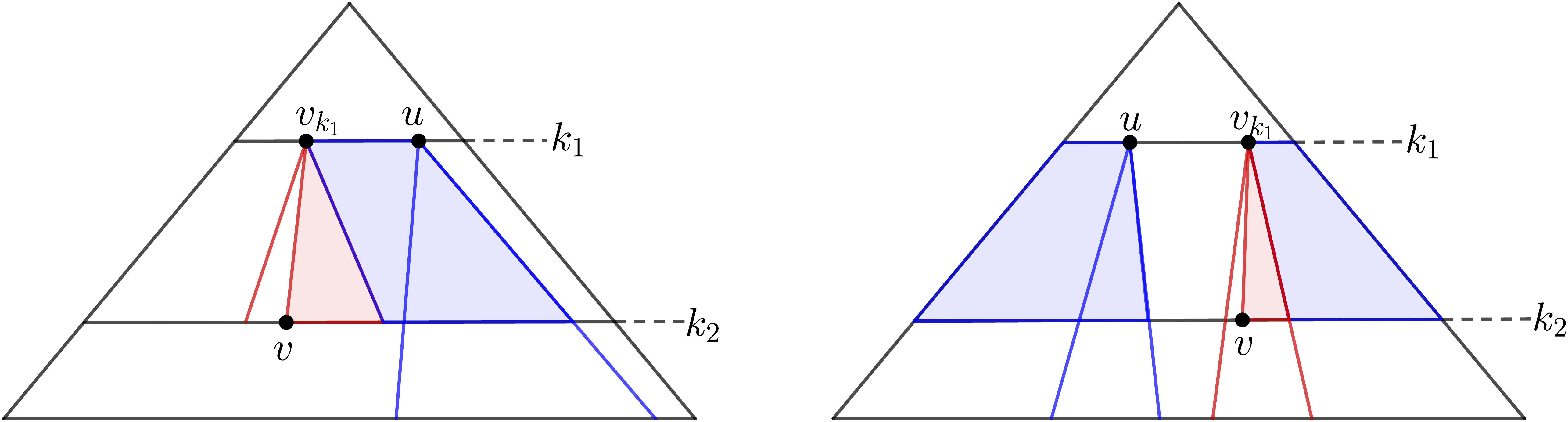}
    \caption{The union of the two shaded regions capture the nodes which preserve the distance in Lemma \ref{l:proper-preserve}. The left is the case when $\qpost(u) > \qpost(v)$, and the right is the case when $\qpost(u) < \qpost(v)$. The blue region represents complete subtrees that are included, while the red represents the nodes to the right of the path to the root from $v$, used in subsection \ref{s:improved-ds}. Note the nodes on level $k_1$ to the left of $u$ and on level $k_2$ to the right of $v$ are not included.}
    \label{f:preserve}
\end{figure}

\begin{lemma}
	\label{l:proper-preserve}
	Let $u< w < v$ be 3 nodes in a proper interval graph. Suppose that $\qpost(u) < \qpost(v)$, then $w$ is on a shortest path exactly when either $\qpost(w) < \qpost(u)$ or $\qpost(w) > \qpost(v)$ (that is $w$ preserves the distance). If $\qpost(u) > \qpost(v)$, then $w$ preserves the distance exactly when $\qpost(v) < \qpost(w) < \qpost(u)$.
	Furthermore, if $w$ does not preserve the distance, then the path passing through $w$ increases the distance by 1. That is $\qdist(u,v)+1 = \qdist(u,w) + \qdist(w,v)$.
\end{lemma}

\begin{proof}
	First we consider the case when $\qpost(u) < \qpost(v)$. By our previous remark, this implies that $v_{k_1} > u$ (that is $v_{k_1}$ is to the right of $u$ on level $k_1$). Furthermore, this implies that $\qdist(u,v) = k_2-k_1+1$.
	
	For each level $k_1 < k_3 \le k_2$, we consider the nodes in level-order between $v_{k_3}$ and $v_{k_3 + 1}$ and denote them as $V_{k_3} = \{w; v_{k_3} \le w < v_{k_3+1}\}$. These are the nodes that are adjacent to $v_{k_3+1}$ (that are before it in level order) and thus have a distance $k_2-k_3$ from $v$. You can see this as a shortest path produced by our $\qsp$ query would be $w, v_{k_3+1},\ldots,v$. This set contains nodes from two levels in the tree: $k_3$ and $k_3+1$. First consider the nodes $w$ on level $k_3$. These nodes are exactly those that $\qpost(w) \ge \qpost(v_{k_3}) > \qpost(v)$. If we look at a shortest path from $w$ to $u$, we see that the chain we produce $w_{k_1},\ldots,w_{k_3}$, has the property that $w_{k_1} \ge v_{k_1} > u$. This is because as $\qpost(w_{k_3}) > \qpost(v_{k_3})$, and this inequality is preserved as we repeatedly take the parent operation on both chains. Therefore, the distance between $w$ and $u$ is $k_3-k_1+1$, and $w$ preserves the distance.
	
	Next consider the nodes on level $k_3+1$. These are the nodes with $w < v_{k_3+1}$, and equivalently, exactly those on this level that $\qpost(w) < \qpost(v)$. As above, we consider the path from $w$ to $u$, which expands out as the chain $w_{k_1},\ldots,w_{k_3+1}$. In the case that $w_{k_1} \le u$, the distance is $k_3-k_1+1$, and if $w_{k_1} > u$, then the distance is $k_3-k_1+2$. Thus we see that $w$ preserves the distance when $w_{k_1} \le u$. This implies that $\qpost(w) < \qpost(u)$, as either $w_{k_1} = u$ and $u$ is the last node in its subtree by post-order numbers, or $w_{k_1}$ is to the left of $u$ and that relation is preserved down the chain.
	Thus in the set $V_{k_3}$, the nodes that preserve the distance are exactly those with $\qpost(w) < \qpost(u)$ or $\qpost(w) > \qpost(v)$. Taking the union of the $V_{k}$ and we see that this is exactly the condition for a node to preserve the distance.
	Furthermore, the nodes that do not preserve the distance only increase the distance by 1 ($\qdist(u,v) + 1 = \qdist(u,w)+\qdist(v,w)$).
	
	The second case is when $\qpost(v) < \qpost(u)$. In this case, we have $v_{k_1} \le u$ and thus $\qdist(u,v) = k_2-k_1$. Again we consider sets $V_{k_3}$ and the nodes on the levels $k_3$ and $k_3+1$. These nodes have a distance of $k_2-k_3$ to $v$.
	First we consider the nodes on $k_3 + 1$. As in the previous case, the distance must be either $k_3-k_1 +1$ or $k_3-k_1+2$, but in either case, we cannot preserve the distance. In fact, since $\qpost(v_{k_3+1}) \le \qpost(u)$, we see that $w_{k_1} \le v_{k_1} \le u$, hence the distance must actually be $k_3-k_1+1$.
	
	Next we consider the nodes on level $k_3$. Expanding out the path, we have two cases: either $w_{k_1} \le u$ or $w_{k_1} > u$. In the first case, the distance is $k_3-k_1$ and these nodes preserve the distance, and the second case the distance is $k_3-k_1+1$. The condition for $w_{k_1} \le u$ is $\qpost(w) < \qpost(u)$, as in this case, either $w$ is in the subtree rooted at $u$ or in the subtree rooted at a node to the left of $u$ on level $k_1$.

	Combining the cases we see that the $w$ preserves the distance exactly when $\qpost(v) < \qpost(w) < \qpost(u)$. And if $w$ does not, it only increase the distance by 1.
\end{proof}

Remark: when $u$ and $v$ are on the same level, then they are adjacent. Thus no nodes can preserve the distance.

See figure \ref{f:preserve} for a pictorial representation of the criteria;
Now we can describe the process of determining the beer distance. The idea is to cover the nodes of the tree with 3 sets, and determine the best possible beer distance using beer vertices in each of the 3 sets. Finally we take the minimum of the 3. We will call the best vertex in each set a \emph{candidate}.

First we note that if either $u,v \in B$, then we do not need to do anything and simply return $\qdist(u,v)$ (or $\qsp(u,v)$).

\textbf{Candidate 1}: The set of nodes is $\{w \in B ; w > v\}$. We claim that the best beer node in this set is the smallest one. To show this, we will use the following lemma.
\begin{lemma}
	\label{l:sort-distance-proper}
	Let $u < v < w$ be 3 nodes in a proper interval graph, then $\qdist(u,v) \le \qdist(u,w)$.
	By symmetry, $\qdist(v,w) \le \qdist(u,w)$.
\end{lemma}
\begin{proof}
	Let the depths be $\qdepth(u) = k_1, \qdepth(v) = k_2, \qdepth(w) = k_3$ with $k_1 \le k_2 \le k_3$. We consider the chain from $w$: $w_{k_1},\ldots, w_{k_3} = w$. Since $v < w$, there exists an index $i$ such that $w_{i-1} < v \le w_{i}$. and $i \le w_3$. Hence we may create the path starting from $v$ as $v \rightarrow w_{i-1} \rightarrow w_{i-2} \cdots w_{k_1}$, which is a path to $u$ of at most the length as the path from $w$. The result follows.
\end{proof}

The node in $\{w\in B; w > v\}$ that minimizes the value of $\qdist(u,w)+\qdist(v,w)$ is of course the $w$ of minimal index.

\textbf{Candidate 2}: The set of nodes $\{w\in B; w < u\}$. We take the largest node in the set as the candidate using Lemma \ref{l:sort-distance-proper}.

\textbf{Candidate 3}: The set of nodes $\{w\in B; u<w<v\}$. By Lemma \ref{l:proper-preserve}, we need to determine whether there exists a node such that either $\qpost(u) < \qpost(w) < \qpost(v)$ or ($\qpost(w) < \qpost(u)$ or $\qpost(w) > \qpost(v)$), depending on $\qpost(u) < \qpost(v)$. If such a node exists, then it is the candidate, with distance $\qdist(u,v)$. If no such node exists, we may take any node as the candidate, with distance $\qdist(u,v)+1$.

\subsection{First Data Structure For Beer Distance}
\label{ss:ds-proper-1}
Here we discuss how to use the previous results to create a data structure for the queries. In this subsection, we discuss a relatively simple data structure which has decent run times. However, the space is dependent on $|B|$ the number of beer nodes and in the case that $|B| = \Theta(n)$ is large, the space bound is also unacceptably large. In the next subsection, we will show how to remove this dependence on $|B|$ at the cost of slightly worse run times.

To store the beer vertices, we store a bit vector $B$ of length $n$ so that $B[i] = 1$ if the $i$th vertex is a beer vertex. This uses $n+o(n)$ bits of space.
As above we will assume that neither $u$ nor $v$ are beer vertices (and we can check by looking at $B[u],B[v]$).

\textbf{Candidate 1}: We use $\qrank(v)$ to find how many beer nodes are up to $v$. The smallest beer node that is larger than $v$ can be found using $\qselect(\qrank(v)+1)$.

\textbf{Candidate 2}: Similarly to candidate 1, we find it by $\qselect(\qrank(u))$.

\textbf{Candidate 3}: For every beer node $w$, we store it in a 2D range emptiness data structure using the coordinates $(w,\qpost(w))$ (by our notation $w$ is just the level-order number of the node $w$). In the case that $\qpost(u) < \qpost(v)$, we need the nodes such that $\qpost(w) < \qpost(u)$ or $\qpost(w) > \qpost(v)$ and $u < w < v$. This is translated to the rectangles $(u,v)\times (-\infty,\qpost(u))$ and $(u,v)\times (\qpost(v),\infty)$.

For the second case when $\qpost(u) > \qpost(v)$, we need the nodes that $\qpost(v) < \qpost(w) <\qpost(u)$. This is the rectangle $(u,v)\times (\qpost(v),\qpost(u))$.

This suffices to determine the distance. To list out the path, we first determine which candidate to use. Candidates 1 and 2 can simply list out the path using two $\qsp$ queries. For candidate 3, we list out the path between $u,v$ one step at a time, and, at each step, we consider the set $V_{k}$, which is an interval in level order. We note that the nodes preserving the distance is a prefix of this interval. Thus we find the first beer vertex in $V_k$, and check if it preserves the distance, if so add it to the path and list out the path from there. Otherwise, we continue to the next level. 
Furthermore, as listing out the path for Candidate 3 is at most $O(\qdist(u,v))$ time, we may do this whenever $\qdist(u,v) = O(\log^\epsilon n)$ rather than spending the time on the orthogonal range search in the distance query.
Thus we have the following theorem:
\begin{theorem}
	A beer proper interval graph $G$ can be represented using $3n + o(n) + O(|B|\log n)$ bits to support the interval graph queries plus $\qbsp$ in $O(1)$ time per vertex on the path and $\qbdist(u,v)$ in $O(\min(\log^\epsilon n,\qdist(u,v)))$ time.
	If we increase the extra space to $O(|B|\log n\log\log n)$ bits, we may support $\qbdist(u,v)$ in $O(\min(\log \log n,\qdist(u,v)))$ time instead.
\end{theorem}
\begin{proof}
	The space required is the distance tree of \cite{DBLP:conf/isaac/0001MNWW20}, a single bit vector for the beer nodes, and a single 2D range emptiness data structure.
\end{proof}

We note that if there are many beer vertices, so that $|B| \in \Theta(n)$, then the space usage would be $\Theta(n\log n)$. However if $|B|$ were small (ex. $|B| = O(n/\log^2 n)$), then this data structure will suffice.

\subsection{Improved Data Structure for Beer Distance}
\label{s:improved-ds}
Here we show how to improve the space usage of our specialized ranged emptiness query, so that it no longer has any dependence on $|B|$. Since we have the tree structure, the range emptiness can be reduced to checking whether a certain set of tree nodes have any beer nodes in them. In particular, as seen from the previous subsection, we need to support the following rectangles using $o(n)$ bits: 1) $(u,v)\times (-\infty,\qpost(u))$, 2) $(u,v)\times (\qpost(v),\infty)$, and 3) $(u,v)\times (\qpost(v),\qpost(u))$.
We will call these type 1,2 and 3 rectangles.

To make our notation cleaner, we will use the depth of a node in the first coordinate of a rectangle. This means to include all the nodes on that level. To accomplish this, we simply find the first node on that level (in $O(1)$ time) and substitute its level-order number as the value to be used in the rectangle; similarly use the last node of a level for the right end point of the rectangle.

Fix $\Delta = \omega(1)$, and choose an index $1 \le i \le \Delta$ such that the number of nodes on levels $k = i \mod \Delta$ is minimized. We will call these levels \emph{selected levels}, and the nodes on them \emph{selected nodes}. Thus the number of selected nodes is $O(n/\Delta)$ by the pigeonhole principle. For each of these nodes, consider the subtree rooted at them that extends down to the next select level. We will build the contracted tree $T'$ with these subtree as nodes, and the appropriate edges. A node in $T'$ is a beer node if at least one of the original nodes in the corresponding subtree \emph{except the root} is a beer node.

We use a bit vector to store which nodes in level order are selected. As the number of nodes is $O(n/\Delta) = o(n)$, this compressed bit-vector uses $o(n)$ bits of space. We store the contracted tree $T'$ succinctly, using $2n/\Delta + o(n) = o(n) $ bits. The bit vector to mark which nodes of $T'$ are beer nodes is also $n/\Delta = o(n)$ bits. Thus the total space for our contracted tree is $o(n)$ bits.

To support these rectangles, we first reduce the general case to one where both $u,v$ are on a selected level. 
\begin{lemma}
\label{l:reduce-to-selected}
    We may assume that the inputs $u,v$ are on selected levels at the cost of $O(\Delta)$ extra time.
\end{lemma}

\begin{proof} 
For $v$, we move up the tree using parent as in the $\qbsp$ query. On level $k$, we take the all the nodes on that level, and find 1) the first beer node, 2) the first beer node after $v_k$. If the first beer node $w$ has the property that $\qpost(w) < \qpost(u)$, we may answer the type 1 rectangle query as true immediately. If there exists a beer node after $v_k$, we may answer the type 2 query as true immediately. If the beer node $w$ after $v_k$ has the property that $\qpost(w) < \qpost(u)$, then we may answer the type 3 query as true immediately. Thus we assume we find no beer nodes that will allow us to answer the query immediately. Since each step takes $O(1)$ time, it takes $O(\Delta)$ time to reach a node that is on a selected level.

For $u$, we move down the tree $T'$. The essence of the rectangles is that for each level, we wish to split the nodes on that level into 2: those with post-order numbers less than or equal to $u$, and those greater. Thus as we move down the tree $T'$ we wish to find the node that splits the levels in the same way as $u$.

By the properties of post-order traversal, the node on the next level that has this property is the largest post-order numbered node less than $u$. If $u$ has any children, this is the last child of $u$. If $u$ has no children, this is the last child of the previous internal node (in level-order) from $u$. As shown by He et al. \cite{DBLP:conf/isaac/0001MNWW20} this is the largest neighbour of $u$. We will call these nodes $u_k$ for the node on level $k$. In the case that there are no nodes that satisfy the criteria (that is every node on the next level have a post-order number larger than $u$), then as no node to the left of $u$ has any children, we must necessarily have that type 1 and type 3 rectangles are empty. As type 2 rectangles do not use $\qpost(u)$ and only its depth, we may choose any node on the closest selected level.

As we descend down the tree, we again take all the nodes on the level, and find 1) the first beer node before $u$ - or equivalently, the first beer node before the node $u_k$, 2) the last beer node. For the first kind $w$, if one exists, we may answer type 1 rectangles as true. We also check that it satisfies $\qpost(w) > \qpost(v)$, and if so answer type 3 rectangles as true. The second kind, we check that $\qpost(w) > \qpost(v)$ and if so, answer type 2 rectangles as true.
\end{proof}

We will now assume that both $u,v$ are on selected levels. To deal with these queries, we will build the 2D range emptiness query on our contracted tree $T'$ and the nodes on the selected levels in the same way. We wish to convert as much of the query to the 2D range emptiness on the contracted tree as possible. We will denote the corresponding node in the contracted tree to $u$ by $u'$.

\textbf{Type 1 rectangles:} we convert $(u,v)\times (-\infty,\qpost(u))$ to \linebreak $[\qdepth(u'),\qdepth(v')) \times (-\infty, \qpost(u'))$ in the contracted tree and the same rectangle \linebreak $(u,v)\times (-\infty,\qpost(u))$ in the selected nodes. We note that in $T'$, we exclude all nodes on $\qdepth(v')$ since in $T$ this includes only the nodes on that level, but in $T'$ this would include the subtrees as well, which extend down. We also change the left endpoint so that we include the subtrees to the left of $u$, but as we exclude the roots from those subtrees (in our decision to mark them as beer nodes or not), we do not include more nodes in our search than required.

\textbf{Type 2 rectangles:} In the type 2 rectangles, we note that unfortunately, the rectangle does not contain all the nodes in the subtree $v_{\qdepth(u)}$, only those to the right of the path to $v$. It does however include the entire subtrees of all the nodes to the right of $v_{\qdepth(u)}$. To handle these complete subtrees (and exclude the subtree rooted at $v_{\qdepth(u)}$) we use the rectangle $[\qdepth(u'),\qdepth(v'))\times (\qpost(v_{\qdepth(u')}),\infty)$. Of course we may find $v_{\qdepth(u)}$ using level-ancestor. We again handle the selected nodes using the same rectangle on them $(u,v)\times (\qpost(v),\infty)$.

Finally, we need to handle the the nodes in the subtree rooted at $v_{\qdepth(u)}$, to the right of the path to $v$ and above the level of $v$. We start with the entire subtree of $v_{\qdepth(u)}$, whose nodes are an interval in post-order. Then nodes $w$ with $\qpost(w) > \qpost(v)$ are exactly the ones we want, except that all the subtrees to the right of $v$ extend down to the bottom of the tree, rather than being cut off at the depth of $v$. Thus we need to handle them as well. The way to do this in encoded in the lemma below, where $k_1 = \qdepth(u)$.

\begin{lemma}
\label{l:emptiness-to-right}
    Let $v$ be a node in a tree $T$ at depth $k_2$ and $v_{k_1}$ be the ancestor of $v$ at depth $k_1$, where both $k_1,k_2$ are selected levels with a fixed $\Delta > 0$. Let $T'$ be the contracted tree as defined above. Then we are able to answer the query: does the rectangle $(v_{k_1},v)\times (\qpost(v),\infty)$ in either:
    
    $O(n/\Delta\log n) + o(n)$ additional space and $O(\log n)$ time.
    
    $O(n/\Delta\log n) + n + o(n)$ additional space and $O(\log\log n)$ time.
    
    Here we do not count the space taken by the tree $T$.
\end{lemma}

\begin{proof} 
We count number of beer nodes in this rectangle and if the count is 0, return false, otherwise return true.

We store a bit vector $P$, where $P[i] = 1$ if the $i$th node in post-order is a beer node. The number of beer nodes in the subtree of $v_{k_1}$ to the right of $v$ using two rank operations at $\qpost(v_{k_1})$ and $\qpost(v)$. Finally we need to subtract off the number of beer nodes below the subtrees rooted at the selected nodes to the right of $v$. To do this, at each selected node $x$, we store the total number of beer nodes in the subtrees to all the selected nodes on the same level $\qdepth(x)$ to the left of $x$ (including $x$). The number we need to compute is the difference in the number of beer nodes at subtrees to left of $v$ and the last node on $\qdepth(v)$ that is a descendant of $v_{k_1}$. The space required to store the number of beer nodes in these subtrees is $O((n/\Delta)\cdot\log n)$. 

We note that normally, to explicitly store $P$, we need $n + o(n)$ bits. As the beer node are stored in $B$ and we can convert between the indices of $B$ and $P$ in constant time, we may forgo storing the bit vector itself, and only store the auxiliary information. Whenever we need a bit of $P$, we convert the post order number to level-order and use our level order bitvector $B$ instead. As the $\qrank$ operation is $O(1)$ we thus need at most $O(\log n)$ bits from the vector $P$ (this occurs when we need $O(\log n)$ contiguous bits as the key into a lookup table). Thus we may implicitly store $P$ using only $o(n)$ bits, at the cost of $O(\log n)$ $\qrank$ query time.

To compute the last node at depth $v$ and is a descendant of $v_{\qdepth(u)}$, we will store the post order numbers of nodes in a predecessor data structure. For each selected level, we store a predecessor structure containing the post order numbers of the selected nodes at that level. The node in question is found by $\qpred(\qpost(v_{k_1}))$ on the data structure containing the nodes at level $\qdepth(v)$. As the number of selected nodes in total is $O(n/\Delta)$, the total space cost of all the predecessor structures is $O((n/\Delta)\cdot \log n)$ bits. The time complexity is $O(\log\log n)$.

Combining these numbers we obtain the number of beer nodes in the given rectangle.

Thus if we store $P$ explicitly, then the space required is $O((n/\Delta)\cdot\log n) + n + o(n)$ with time $O(\log\log n)$. 

If we do not store $P$ explicitly, then the space required is $O((n/\Delta)\cdot\log n) + o(n)$ with time $O(\log n)$.
\end{proof}

\textbf{Type 3 Rectangles:} Type 3 rectangles are similar to type 2 rectangles. As above, we use the same rectangle $(u,v)\times (\qpost(v),\qpost(u))$ for the selected nodes. We again use the rectangle 
$[\qdepth(u'),\qdepth(v'))\times (\qpost(v_{\qdepth(u')}),\qpost(u')]$ to capture the complete subtrees that we wish to use. The incomplete subtree is exactly the same as in type 2, so we are able to apply lemma \ref{l:emptiness-to-right}.

Thus putting everything together we have the following theorem:
\begin{theorem}
	Let $G$ be a beer proper interval graph. Fix $\Delta$, then $G$ can be represented using $3n+o(n) + O((n/\Delta) \cdot \log n)$ bits and can support $\qadj,\qdeg,\qnb,\qdist$ in $O(1)$ time, $\qsp,\qbsp$ in $O(1)$ time per vertex on the path and $\qbdist$ in $O(\Delta + \log n)$ time.
	
	In particular, if we take $\Delta = \log n$, then the space is $O(n)$ with time $O(\log n)$ and if we take $\Delta = f(n)\log n$ for some $f(n) = \omega(1)$, then the space is $3n + o(n)$ and the time is $O(f(n)\log n)$.
\end{theorem}

If we wish to further our trade off and improve the time, we must explicitly store $P$ as in the proof lemma \ref{l:emptiness-to-right}, and use the space inefficient (but time efficient) range query data structures. Thus we obtain:

\begin{theorem}
	Let $G$ be a beer proper interval graph. Fix $\Delta$, then $G$ can be represented using $4n+o(n) + O((n/\Delta) \cdot \log n\log\log n)$ bits and can support $\qadj,\qdeg,\qnb,\qdist$ in $O(1)$ time, $\qsp,\qbsp$ in $O(1)$ time per vertex on the path and $\qbdist$ in $O(\Delta + \log\log n)$ time.
	
	In particular, if we take $\Delta = \log\log n$, then the space is $O(n\log n)$ with time $O(\log\log n)$.
\end{theorem}
\section{Beer Paths in Interval Graphs}
\label{s:ds-interval}
In this section, we study how to compute and construct data structures for beer paths and beer distances. We will use the the data structure of Acan et al. \cite{DBLP:journals/algorithmica/AcanCJS21} as a black box, and work with the distance tree $T$ of He et al. \cite{DBLP:conf/isaac/0001MNWW20}. The major difference between interval graphs and proper interval graphs is how adjacency can be checked. In a proper interval graph, for a vertex $v$, and its parent $p$ in the distance tree, $v$ is adjacent to every vertex between $v$ and $p$. However, in interval graphs, this is not the case, and depending on the graph structure, any of those vertices can be adjacent or not adjacent to $v$.

\subsection{Calculating Beer Distance}

As in proper interval graphs, we begin by investigating the conditions in which nodes \emph{preserve} the distance. For a node $w$, we say that $w$ is $+k$ distance if $\qdist(u,v) +k = \qdist(u,w) + \qdist(v,w)$ (and thus preserving the distance is equivalent to being +0 the distance). So, using $w$ as a beer node will add $k$ to the optimal (non-beer path) distance. To do this, we will add one more condition to that of Lemma \ref{l:proper-preserve}. 

Let $u$ be a node in $T$. As shown in the proof of lemma \ref{l:reduce-to-selected}, the node on the next level that splits it in the same way as $u$ is the largest neighbour of $u$. We denote this by $\qlast(u)$. 
For example, in example \ref{e:proper-interval}, the largest neighbour of the node 8, is the node 13, as 8 is adjacent to node 13, but not node 14. Thus $\qlast(8) = 13$.

\begin{lemma}
	\label{l:interval-preserve}
	Let $u < v$ be vertices in a beer interval graph $G$ with depths $\qdepth(u) = k_1 \le k_2=\qdepth(v)$. Consider the nodes $u < w < v$. Then we have the following two criteria:
	\begin{itemize}
		\item If $\qpost(\qlast(u)) < \qpost(v)$, then either $\qpost(w) > \qpost(v)$ or $\qpost(w) < \qpost(\qlast(u))$. If $\qpost(\qlast(u)) > \qpost(v)$, then $\qpost(v)<\qpost(w)<\qpost(\qlast(u))$.
		\item If $\qpost(w) < \qpost(v)$, then $\qpost(\qlast(w)) > \qpost(v)$ or $\qpost(\qlast(w)) < \qpost(w)$. If $\qpost(w) > \qpost(v)$, then $\qpost(v)<\qpost(\qlast(w))<\qpost(w)$.
	\end{itemize}
	If $w$ satisfies both criteria, then $w$ preserves the distance. If $w$ satisfies one criteria, then $w$ is +1 the distance and if $w$ satisfies neither criteria, then $w$ is +2 the distance.
\end{lemma}

\begin{proof} 
	Consider the path to the root from $v$: $v_1,\ldots , v_{k_2}=v$, and let $k$ be the node where $v_k \le \qlast(u) < v_{k+1}$.
	As in Lemma \ref{l:proper-preserve}, let the nodes $V_i = \{u < w < v; v_i \le w < v_{i+1}\}$, which we will call slices.
	
	We wish to split $V_i$ based on the nodes distances to $u$. We will show that $V_i = V_i^+ \cup V_i^-$ where $w \in V_i^+$ if $\qdist(w,u) = \qdist(v_i,u)$ and $w \in V_i^-$ if $\qdist(w,u) = \qdist(v_{i+1},u)$.
	
	By the distance algorithm, suppose that $x \in V_i^+$, then any children of $x$, $c$, is in $V_{i+1}^+$. We can see this as $\qdist(c,u) = \qdist(x,u) + 1 = \qdist(v_i,u) + 1 = \qdist(v_{i+1}, u)$.
	
	Suppose that $\qdepth(\qlast(u)) = k$, so that $\qlast(u)$ is on the same level as $v_k$. Then $V_i^+ = \{w \in V_i; \qdepth(w) = i, \qpost(w) \le \qpost(\qlast(u))\}$ and $V_i^-$ is the remaining nodes. Otherwise, if $\qlast(u)$ is on the same level a $v_{k+1}$, then $V_i^+ = \{w \in V_i; \qdepth(w) = i \text{ or } \qpost(w) < \qpost(\qlast(u))\}$.
	
	We show this by induction. Consider the first slice that is non-empty, that is $V_{k}$. In the first case that $\qdepth(\qlast(u)) = k$, the nodes in $ \{w \in V_k; \qdepth(w) = k, \qpost(w) \le \qpost(\qlast(u))\}$, are those adjacent to $u$, and everything else is non-adjacent, but are adjacent to $v_k$, hence all other nodes have a distance of 2.
	
	On the other hand, if $\qdepth(\qlast(u)) = k+1$, then the nodes $ \{w \in V_k; \qdepth(w) = k \text{ or } \qpost(w) < \qpost(\qlast(u))\}$ are those that are adjacent to $u$ and have distance 1, and the rest have distance 2.
	
	Since if a node $x$ satisfies $\qpost(x) \le \qpost(\qlast(u))$ if and only if any children of $x$, $c$ also satisfies $\qpost(c) \le \qpost(\qlast(u))$, we have that $V_i^+ = \{w \in V_i; \qdepth(w) = i, \qpost(w) \le \qpost(\qlast(u))\}$ or $V_i^+ = \{w \in V_i; \qdepth(w) = i \text{ or } \qpost(w) < \qpost(\qlast(u))\}$.
	
	For a node $w$, consider the interval $(w,\qlast(w)]$. We are interested in whether this interval contains one of the nodes $v_i$. We use the open interval on the left because if $w = v_k$ for some $k$, then $\qlast(w) \ge v_{k+1}$, so it still contains one of the $v_i$'s, and by doing so, we make sure that exactly one $v_i$ can be contained in the interval.
	
	First suppose that $(w,\qlast(w)]$ contains one such node, say $v_{k-1}$. Then $w \in V_k$ and $w$ is adjacent to $v_{k+1}$.
	Thus $\qdist(v_k,v) = \qdist(w,v)$ by the distance algorithm. Furthermore, if $w \in V_k^+$, then $\qdist(u,w) = \qdist(v_k,u)$, so that $\qdist(u,v) = \qdist(u,w)+\qdist(v,w)$.
	On the other hand, if $w \in V_k^-$ instead, then $\qdist(u,v) + 1 = \qdist(u,w)+\qdist(v,w)$.
	
	Next suppose that $(w,\qlast(w))$ does not contain any $v_i$, and thus, for some $k$, $v_k < w < v_{k+1}$. Since $w$ is not adjacent to $v_{k+1}$, then we have $\qdist(w,v) = \qdist(v_k,v)+1$. Therefore, if $w \in V_k^+$, then $\qdist(u,v)+1 = \qdist(u,w) + \qdist(v,w)$ and if $w \in V_k^-$, then $\qdist(u,v)+2 = \qdist(u,w) + \qdist(v,w)$.
	
	Finally we wish to write the criteria that the interval $(w,\qlast(w)]$ contains one of the $v_i$ in a more computable form. In the first case that $\qpost(w) < \qpost(v)$, so that $w$ is to the left of the path, we need that either $\qpost(\qlast(w)) > \qpost(v)$ or $\qpost(\qlast(w)) < \qpost(w)$. The first case capture when the interval stays on the same level in the tree, and the second captures when the interval wraps to the next. In the second case that $\qpost(w) \ge \qpost(v)$, we need that $\qpost(v) < \qpost(\qlast(w)) < \qpost(w)$.
\end{proof}

Again, we will find several sets that cover the beer nodes, and argue about the optimal beer node in these sets. We then take the minimum distance of these candidates. As before, we will assume that neither $u$ nor $v$ are beer nodes.

\textbf{Candidate 1}: This is the same as the proper interval graphs: $\{w\in B; w > v\} = \{w \in B; l_w > l_v\}$. We again claim that the best beer node is the smallest one. It turns out the exact same proof of Lemma \ref{l:sort-distance-proper} will work here.

\textbf{Candidate 2}: We wish to use the symmetric set of Candidate 1. Unfortunately, in interval graphs, this is not as simple. The set is $\{w \in B; r_w < r_u\}$. We note that this condition and the proper interval graph non-nesting condition gives $l_w < l_u$ so that $w < u$, and thus this is the right analogous set to consider.
We claim that the best node is the node with the largest $r_w$ in this set. This can be seen by reflecting the intervals of the vertices - equivalent to sorting them by the right endpoints instead. In this reflected graph, apply Lemma \ref{l:sort-distance-proper}, and the result follows.

\textbf{Candidate 3}: The nodes $\{w\in B; u < w < v\}$. By Lemma \ref{l:interval-preserve}, we can obtain the distance of the best beer node using the criteria in the lemma.

\textbf{Candidate 4}: The left over nodes. The nodes that do not belong to the previous candidate sets are $w$ with: $l_w < l_v$ and $r_w > r_u$ and $l_w < l_u$. Thus these are the nodes with $l_w < l_u < r_w$, so they are adjacent to $u$. Formally, this is the set: $\{w \in B; w < u, r_w > r_u\}$. As these nodes are adjacent to $u$, all we need to check is their distance to $v$.

First consider the path from $u$ to $v$. By the distance algorithm, we have the path to the root from $v$: $v_{1},\ldots,v_{k_2} = v$, And suppose that $k$ is the index such that $v_k \le u < v_{k+1}$. The end of the path could look like either $u, v_{k+1}, \ldots$ or $u,v_k, v_{k+1},\ldots$, depending on whether $v_{k+1}$ is adjacent to $u$ or not.

First assume that $v_{k+1}$ is not adjacent to $u$, then for any possible candidate $w$, if $w$ were adjacent to $v_{k+1}$ (that is $\qlast(w) \ge v_{k+1}$), then $w$ preserves the distance (as $\qdist(w,v) = \qdist(u,v)-1$). Otherwise, as $r_w > r_u > l_{v_k}$, $w$ is adjacent to $v_k$ and hence $\qdist(w,v) = \qdist(u,v)$ and $w$ is +1 the distance.

Next assume that $v_{k+1}$ is adjacent to $u$. Then again for any candidate $w$, $l_w < l_u < l_{v_k+1} < r_u < r_w$, so $w$ is adjacent to $v_{k+1}$ and $w$ is $+1$ the distance. We note that $w$ cannot be adjacent to $v_{k+2}$ as in this case we would contradict that fact that $v_{k+1}$ is the smallest node adjacent to $v_{k+2}$, by the parent relationship in $T$.

Finally we note that the only property of $w$ that we used is that $w$ is adjacent to $u$, and that if $x > u$ is adjacent to $u$, then $w$ is also adjacent to $x$ and hence we may relax the set to $\{w \in B; w < u,\qlast(w) \ge \qlast(u)\}$.

\subsection{Data Structure for Beer Distance}

We discuss how to use the previous results to create a data structure for the queries. We use the data structure of He et al. \cite{DBLP:conf/isaac/0001MNWW20} which supports the interval graph queries in optimal time. This uses $n\log n + O(n)$ bits of space. We note that this data structure itself builds upon the data structure of Acan et al. \cite{DBLP:journals/algorithmica/AcanCJS21}.
We store a bit vector $B$ as before, which stores which nodes are beer nodes in level-order. This take $n + o(n)$ bits.

\textbf{Candidate 1}, we handle this in exactly the same way as in proper interval graphs.

\textbf{Candidate 2}, we need to be able to find nodes in the mirrored graph.
\begin{lemma}
\label{l:cand2-interval}
    We can find the desired node in the mirrored graph using $n\log n + O(n)$ bits in $O(1)$ time.
\end{lemma}

\begin{proof}
To find the appropriate node, we need to be able to implement the following steps:
\begin{itemize}
\item For a node $u$, what is its index in the mirrored graph?
\item For a node $u$ in the mirrored graph, what is the smallest beer node larger than it?
\item For a node $u$ in the mirrored graph, what is its index in the original graph?
\end{itemize}

For point 1, for a node $u$, we get its interval right endpoint $r_u$. In the data structure of Acan et al., we have a length $2n$ bit vector, which stores whether the endpoint at $i$ is a right endpoint or a left endpoint. We use the rank operation to find how many intervals $i$ have right endpoints less than $r_u$, and thus $u$ is the $n-i$th node in the mirrored graph.

For point 2, we store the analogous bitvector $B_R$ for the mirrored graph, which says whether vertex $i$ in the mirrored graph is a beer node, and use it to find the appropriate beer node. This takes $n+ o(n)$ bits.

For point 3, we store a map mapping vertex $i$ in the mirrored graph to the corresponding node in the original graph. This takes $n\log n$ bits.

In total, this takes $n\log n + O(n)$ bits, and all of these operations are $O(1)$ time.
\end{proof}

\textbf{Candidate 3} We are able to handle this using 3D 5-sided orthogonal range emptiness data structures.

\begin{lemma}
\label{l:cand3-interval}
    We can check if a beer node satisfies the criteria of lemma \ref{l:interval-preserve} using a constant number of 3D 5-sided orthogonal range emptiness data structures. Thus the space/time requirements are either $O(n\log n)$ space and $\log^\varepsilon n$ time or $O(n\log n\log\log n)$ space and $\log\log n$ time.
\end{lemma}

\begin{proof}
We will again use the range emptiness method described for proper interval graphs.
However, as we need an additional criteria, we also need another dimension in the grid to store this criteria.
Thus for each beer node $b$, we store the following point in a 3D table: $(b,\qpost(b), \qpost(\qlast(b)))$.

We note that it is difficult to express the condition $\qpost(\qlast(b)) < \qpost(b)$, as it depends on the values of the node being filtered. Furthermore in the second point, we would need 6-sided rectangles. To alleviate this, we create 2 tables: 1 for beer nodes whose intervals stay on the same level (that is $\qpost(\qlast(b)) > \qpost(b))$) and one for those that wrap to the next level. We will names these $R_1$ and $R_2$.

To do this, we will first query the criteria separately, and then query them together. If either of the criteria returns a positive, then we know that the optimal beer node is either +0 or +1 the distance, depending on the result of the joint query. If neither the criteria return a candidate, then the optimal node is +2 the distance(we will need to check using $B$ that there is a beer node in this range to use).

For the first criteria, directly translating the condition from Lemma \ref{l:interval-preserve}, we obtain the following rectangles.

\begin{itemize}
\item $(u,v)\times (-\infty,\qpost(\qlast(u))) \times [-\infty,\infty]$,

\item $(u,v)\times (\qpost(v),\infty) \times [-\infty,\infty]$, and

\item $(u,v)\times (\qpost(v),\qpost(\qlast(u))) \times [-\infty,\infty]$.
\end{itemize}

As we have split the beer nodes into two data structures $R_1$ and $R_2$, we need to do the query on both.

To check only the second criteria, we have the following rectangles:
Does $(u,v)\times [-\infty,\qpost(v)]\times [\qpost(v),\infty]$ in $R_1$ or  $(u,v)\times [-\infty,\qpost(v)]\times [-\infty,\infty]$ in $R_2$ contain any nodes?

Does $(u,v)\times [\qpost(v),\infty]\times [\qpost(v),\infty]$ in $R_2$ contain any nodes?

Finally to check both criteria at the same time, we take the intersection of the rectangles from the two separate criteria. As the intersection of rectangles are rectangles, with potentially more sides, we may do this. Since the third coordinate is always $[-\infty,\infty]$ in criteria 1, and it is open ended on at least 1 side in criteria 2, we see that any intersection is at most a 5 sided rectangle.
\end{proof}

\textbf{Candidate 4} We will use the following lemma:
\begin{lemma}
    \label{l:cand4-interval}
    We can convert the criteria of candidate 4 into a constant number of 5-sided rectangles.
\end{lemma}

\begin{proof}
The nodes we are interested in are those with $w < u$ and $\qlast(w) \ge \qlast(u)$. Let $p(u)$ denote the parent of $u$ in $T$. All such $w$ are adjacent to $\qlast(u)$ and thus $p(\qlast(u)) \le w < u$.  First we find whether this Candidate set is empty or not. As we are checking the condition $\qlast(u) \in [w,\qlast(w)]$, this will be similar to the second criteria of candidate 3. The rectangles are:

\begin{itemize}
\item $[p(\qlast(u)),u)\times [-\infty,\qpost(\qlast(u))]\times [\qpost(\qlast(u)),\infty]$ in $R_1$,

\item $[p(\qlast(u)),u)\times [-\infty,\qpost(\qlast(u))]\times [-\infty,\infty]$ in $R_2$,

\item $[p(\qlast(u)),u)\times [\qpost(\qlast(u)),\infty]\times [\qpost(\qlast(u)),\infty]$ in $R_2$.
\end{itemize}

As described in the previous part, there are two cases, either $v_k$ is adjacent to $u$ or $v_k$ is not adjacent to $u$ (we check this in $O(1)$ time from the distance algorithm).

In the case that $v_k$ is not adjacent to $u$. We wish to find a node $w < u$ such that $\qlast(w) > v_k$ (so that $v_k \in [w,\qlast(w)]$).
Any such $w$ is adjacent to $v_k$ and must satisfy $p(v_k) \le w < v_k$. Thus we replace all instances of $\qlast(u)$ with $v_k$ in the rectangles above.
If a node exists then it preserves the distance, and if no such node is found, then the best possible is $+1$ the distance.

In the case that $v_k$ is adjacent to $u$, we do not need to do anything more, since any $w$ is $+1$ the distance.
\end{proof}

Finally, to handle shortest paths, we use the reporting query rather than the emptiness query. When the reporting query returns the first point, we stop. After we find the best beer node, we list out the path using two $\qsp$ queries.
\begin{theorem}
	Let $G$ be a beer interval graph, with beer nodes $B$. The there exists a data structures using $2n\log n + O(n) + O(|B|\log n)$ bits that supports $\qdeg,\qadj,\qdist$ in $O(1)$ time, $\qnb,\qsp$ in $O(1)$ time per vertex in the path/neighbourhood, $\qbdist$ in $O(\log^\epsilon n)$ time and $\qbsp$ in $O(\log^\epsilon n + d)$ time where $d$ is the distance between the two vertices.
	
	Alternatively, we may increase the space from $O(|B|\log n)$ to $O(|B|\log n\log\log n)$ and replace the $\log^\epsilon n$ in $\qbsp,\qbdist$ with $\log\log n$.
\end{theorem}

\section{Lower Bounds for Beer Interval Graphs}
\label{s:lower-bound}
In this section, we will derive lower bounds for beer interval graphs and beer proper interval graphs. The lower bounds we will derive will be information theoretic, so that for a set of objects $X$, we will need at least $\log |X|$ bits in the worst case to represent any specific object.

First, we note that it is not interesting for beer interval graphs, since adding beer vertices can increase the lower bound by at most $n$ bits. Since the lower bound for interval graphs is already $n\log n - o(n\log n)$ bits, the increase in space to account for the beer vertices is a lower order term and does not impact our data structures.
For proper interval graphs however, the lower bound is $2n$ and thus it is natural to ask whether adding the beer vertices requires the full $n$ bits to store them. That is, is the lower bound for beer proper interval graphs $3n$? In the main result of this section, we will show that it is not necessarily the case, and that if $X$ were the set of beer proper interval graphs, then $\log|X| = n\log(4+2\sqrt{3}) - o(n) \approx 2.9n$.

In our case, we are interested in beer graphs, that is a graph $G$ together with a set $B \subseteq V$ of beer vertices. We will refer to $B$ as a beer vertex pattern. We will say that two beer graphs $(G_1,B_1)$ and $(G_2, B_2)$ are isomorphic (and thus are the same object) if there exists a bijection $f: V(G_1) \mapsto V(G_2)$ such that $(u,v) \in E(G_1) \Leftrightarrow (f(u),f(v)) \in E(G_2)$ and $u \in B_1 \Leftrightarrow f(u) \in B_2$. The first condition is the standard condition for two graphs to be isomorphic and the second condition says that this isomorphism also preserves beer vertices. Thus for two beer graphs to be isomorphic, the underlying graphs must also be isomorphic as well.

\begin{example}
    \label{e:clique}
    Suppose our graph class are cliques, then how many beer cliques are there? On $n$ vertices, there is exactly one underlying graph $G = K_n$ on $n$ vertices that is a clique. Thus it remains to see how many different beer vertex patterns we can have. By definition, if $(K_n, B_1)$ were isomorphic to $(K_n, B_2)$, then there exists an automorphism $\sigma$ of $K_n$ mapping vertices $u \in B_1$ to $\sigma(u) \in B_2$ bijectively, and thus $|B_1| = |B_2|$. Conversely, if $|B_1| = |B_2|$ then there exists a bijection $\sigma$ that maps the elements of $B_1$ to $B_2$ and fixes every other vertex. As the underlying graph is the complete graph $K_n$, this $\sigma$ is also an automorphism of the underlying graph as well.
    Thus $(K_n, B_1)$ is isomorphic to $(K_n, B_2)$ exactly when $|B_1| = |B_2|$. The number of different ways to add beer nodes to a clique on $n$ vertices is thus $n+1$. \qed
\end{example}

As $(G_1, B_1) \cong (G_2, B_2)$ happens only when $G_1 \cong G_2$, it remains to develop the theory to compute the number of beer vertex patterns that are different when given a specific underlying graph $G$. 
Let $Aut(G)$ denote the automorphism group of a graph $G$. We will view $B \subseteq V$ as a vector $B \subseteq 2^n$ on the hypercube (where the $i$-th bit denotes whether the $i$-th vertex belong to the set or not), and $Aut(G)$ as a group that acts on $2^n$. In this lens, two beer vertex patterns $B_1, B_2$ are the same if there exists a group element $\sigma$ mapping $B_1$ to $B_2$, and thus $B_1$ and $B_2$ belong to the same orbit of this group action. The number of different beer vertex patterns is thus the number of orbits $|2^n/Aut(G)|$.
To count the number of orbits, we will use the Polya enumeration theorem \cite{10.1007/BF02546665-polya}, which in its most basic form, states that if we denote $c(\sigma)$ as the number of cycles in $\sigma$ when viewed as a permutation of $V(G)$, $|2^n/Aut(G)| = \frac{1}{|Aut(G)|}\sum_{\sigma \in Aut(G)}2^{c(\sigma)}$.

\subsection{Automorphism Groups of Proper Interval Graphs}


Klavic and Zeman \cite{stacs-klavk_et_al} showed that $Aut(\text{connected PROPER INT}) = Aut(\text{CATERPILLAR})$. A caterpillar graph/tree is a path together with a set of leaves that are adjacent some vertex on the path. In particular, the automorphism group of any particular connected proper interval graph is generated by 2 types of automorphisms. First are automorphisms that swap twin vertices - which corresponds to those that swap the leaves adjacent to the same vertex on the path of a caterpillar graph. In a proper interval graph, twin vertices are those that have the same set of maximal cliques. The second is an automorphism that reverses the proper interval graph, which corresponds to reversing the path of a caterpillar graph. This reversal corresponds to a reversal of the maximal cliques.
Of course, for any particular graph, there may not be any twin vertices, and thus there are no automorphisms of the first type. As for the second type, it can only exist when the number of vertices in the maximal cliques are symmetrical - as the vertices in the first maximal clique are mapped to those in the last maximal clique etc.

We will assume that the maximal cliques are not symmetrical and thus no automorphisms of the second type exists. To see this, we may always desymmetrize the sequence of maximal cliques by adding one  vertex to only the first maximal clique if necessary.

Now suppose that $G$ is a connected proper interval graph. As being twin vertices are an equivalence relation, let $S_1,\ldots, S_h$ be the equivalence classes of twin vertices, that is $u,v \in S_i$ implies that $u,v$ are twins. Let $k_i = |S_i|$ and we will say that vertices which have no twins are in a class of size 1, so that $\sum_i k_i = n$.

\begin{lemma}
\label{l:num-patterns}
    Let $G$ be a proper interval graph with twin vertex classes of sizes $|S_1| = k_1,\ldots, |S_h| = k_h$. Then $|2^n/Aut(G)| \le (k_1+1)(k_2+1)\cdots(k_h+1)$.
    This is an equality in the case that the graph is connected and the maximal cliques are non-symmetrical.
\end{lemma}

\begin{proof}
In the case that the graph is connected and the maximal cliques are non-symmetrical, we have only type 1 automorphisms.

Let $\sigma_i$ be a permutation that permutes only those vertices of $S_i$ and $\sigma_j$ permuting those of $S_j$, then $\sigma_i\sigma_j=\sigma_j\sigma_i$ as $S_i \cap S_j = \emptyset$. Thus we may write $Aut(G) \cong \mathbb{S}_{k_1}\times\cdots\times\mathbb{S}_{k_h}$ where $\mathbb{S}_n$ denotes the symmetric group (set of all permutations) on $n$ elements.

Applying this to Polya enumeration theorem, we obtain that
\begin{align*}
    |2^n/Aut(G)| & = \frac{1}{|Aut(G)|}\sum_{\sigma \in Aut(G)}2^{c(\sigma)}\\
    & = \frac{1}{k_1!\cdot k_2!\cdots k_h!}\sum_{\sigma_1\times\sigma_2\times\cdots\sigma_h\in \mathbb{S}_{k_1}\times\cdots\times\mathbb{S}_{k_h}} 2^{c(\sigma_1\times\sigma_2\times\cdots\sigma_h)}\\
    & = \left(\frac{1}{k_1!}\sum_{\sigma_1\in\mathbb{S}_{k_1}}2^{c(\sigma_1)}\right)\left(\frac{1}{k_2!}\sum_{\sigma_2\in\mathbb{S}_{k_2}}2^{c(\sigma_2)}\right)\cdots\left(\frac{1}{k_h!}\sum_{\sigma_h\in\mathbb{S}_{k_h}}2^{c(\sigma_h)}\right)\\
    & = |2^{k_1}/Aut(K_{k_1})|\cdots|2^{k_h}/Aut(K_{k_h})|\\
    & = (k_1+1)(k_2+1)\cdots(k_h+1)
\end{align*}
The last equality comes from our example \ref{e:clique} dealing with cliques. 

Finally, we note that by the definition of group action, if $H_1 \subset H_2$ are two groups acting on a set $X$, then the orbit of any element $x \in X$ under $H_2$, $H_2\cdot x = \{h\cdot x; h \in H_2\}$ is a superset of that of the orbit under $H_1$. Thus the number of orbits $|X/H_2| \le |X/H_1|$. As the above equation applies exactly to non-symmetric connected proper interval graphs, and dropping the connectedness/symmetric property only increases the automorphism groups (if two connected components are isomorphic, then there is an automorphism that swaps the two connected components), we conclude that if $G$ were a proper interval graph instead, we may say that $|2^n/Aut(G)| \le \Pi_i (k_i + 1)$.

As a further remark, we see that any vertex that is not part of any twins contributes a multiplicative factor of 2 to the above quantity (intuitively, this means that you must store a bit stating whether this vertex is a beer vertex or not) while any vertex that have twin vertices contributes a much smaller term (intuitively, this means that it is necessary to only store only the number of beer vertices using $\log k_i$ bits, furthering our intuition from example \ref{e:clique}). 
\end{proof}
For a proper interval graph, we will also refer to the above quantity as its weight, as the number of beer proper interval graphs would be the weighted sum of proper interval graphs.

\subsection{Lower Bound}

In the section, we will prove a tight lower bound on the number of beer proper interval graph. The main idea is to decompose a Dyck path (which is in more or less a one-to-one correspondence to proper interval graphs) in such a way that it preserves the weights.

To see the one-to-one correspondence, we see that by the interval graph recognition algorithm of Booth and Leuker \cite{DBLP:journals/jcss/BoothL76}, the PQ-tree of the maximal cliques is a single $Q$ node, so that there are at most two ordering of maximal cliques representing each graph (one order and the reverse of that order). As each Dyck path gives a different sequence of maximal cliques, each graph can have at most two Dyck paths representing it.

For a particular proper interval graph, with twin vertex classes of sizes $k_1, \ldots, k_l$, the weight assigned to it is $\Pi_i (k_i + 1)$. Now consider the following blocking scheme for the distance tree associated with the proper interval graph: start at the root and continue in level-order, add the vertices to the block until either: the vertex has a different parent, or the vertex is not a leaf. In this manner, we consider the root as a sibling of its left child.

\begin{lemma}
\label{l:twins}
    In the above blocking scheme, two vertices $u,v$ are in the same block if and only if they are twins.
\end{lemma}

\begin{proof}  
    By the characterization of the neighbourhood of He et. al \cite{DBLP:conf/isaac/0001MNWW20}, we see that the neighbourhood of any vertex $v$ is an interval $[v_1,v_2]$ where $v_1 = $ parent$(v)$ and $v_2$ is the rightmost child of the previous internal node of $v$ in level order.
    
    Thus two vertices $u < v$ are twins exactly when their neighbourhoods coincide, and this neighbourhood is $[v_1,v_2]$. They have the same $v_1$ exactly when they have the same parent. They have the same $v_2$ exactly when the previous internal node is the same, but that means $v$ must be a leaf (as otherwise by definition, $v$ is the previous internal node to $v$, and cannot be the previous internal node to $u$) and furthermore all the vertices between $u$,$v$ are also leaves. Thus by definition, they would all be added to the same block as $u$.
\end{proof}

We may look at this in the same way by replacing the root with a dummy root and dropping the original root as the first child of the dummy root. This blocking scheme is illustrated in Figure \ref{f:blocking}.

\begin{figure}
    \centering
    \includegraphics[scale=0.5]{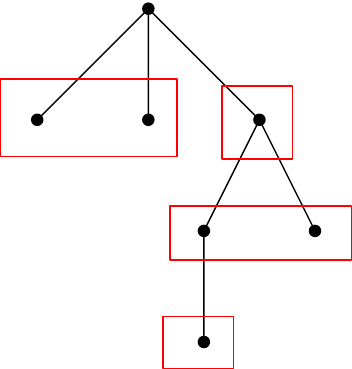}
    \includegraphics[scale=0.5]{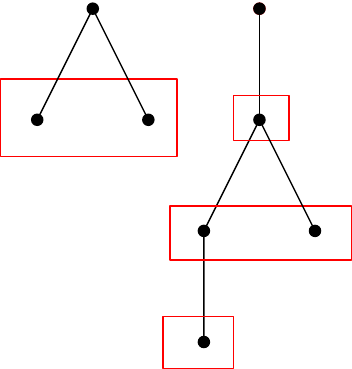}
    \caption{The twin vertex classes in the distance tree, and a way to decompose the tree}
    \label{f:blocking}
\end{figure}
Note that we do not consider the dummy root as part of our blocks and we can also view this as deleting the dummy root and consider the roots of this new forest as siblings. The second is our proposed way to decompose the tree into two trees while preserving all the blocks. If we consider the balanced parenthesis view of the tree (without the dummy root), we see that the sequences are the same $()()|((())())$, but we cut it into two at the $|$. Precisely, $|$ is at the first spot in the sequence such that the excess is 0 and the next two parentheses are $(($. In the language of Dyck paths, this is the first time the path touches the $x$-axis and the next two steps are both up-steps.

Let $C(n)$ be the $n$-th Catalan number and the number of Dyck paths of length $2n$. The above decomposes the path into two subpaths. Let $L(n)$ be the number of paths of length $n$ for the subpath to the left of $|$. and $R(n)$ be the number of paths of length $n$ of the subpath to the right. We will also abuse notation and use $L(n), R(n)$ as the set of Dyck paths of the respective forms. 
Thus we have the recurrence $C(n) = \sum_{k = 0}^n L(k)R(n-k)$.




Now we consider the weighted versions. Let $\bar{C}(n)$ be the sum of all Dyck paths with our weighting system. Similarly for $\bar{L}(n)$ and $\bar{R}(n)$. Because we preserve all the blocks with our split, we have the same recurrence $\bar{C}(n) = \sum_{k = 0}^n \bar{L}(k)\bar{R}(n-k)$, which holds for $n \ge 1$. For $n = 0$, we see that $\bar{L}(0) = 0, \bar{R}(0) = 1$ and $\bar{C}(0) = 1$. Now it remains to compute $\bar{L}(n)$ and $\bar{R}(n)$.

\begin{lemma}
    \label{l:lr}
    $\bar{L}(n) = \displaystyle\sum_{l=0}^n \bar{C}(n-l-1)(l+2) \text{ and } \bar{R}(n) = \bar{C}(n) - \displaystyle\sum_{l=1}^{n}(l+1)\bar{R}(n-l)$.
\end{lemma}

\begin{proof}
    We split a Dyck path in $L(n)$ as above: an irreducible Dyck path followed by a sequence of up-downs. Suppose that the irreducible Dyck path has length $2(n-l)$, then the top level has a block of size $l+1$, as there are $l$ up-downs and the 1 node contributed by the irreducible Dyck path. The remainder of the Dyck path is of length $2(n-l-1)$ and thus contributes $\bar{C}(n-l-1)$.
    
    This is illustrated in Figure \ref{f:lr}. As the first part of the Dyck path is irreducible, it is a rooted tree, and the block at level 1 contains $l+1$ nodes. The remainder of the blocks are exactly the same as those in the first node's subtree.
    
    We again decompose $R(n)$ as all path minus those that begin with up-down. Let $l$ be the number of up-downs that begins the path. These corresponds to leaves that begin the tree and are in a block together. The rest of the path must begin with up-up and is thus a path in $R(n-l)$. Therefore, we have $\bar{R}(n) = \bar{C}(n) - \sum_{l=1}^n(l+1)\bar{R}(n-l)$.
    
    This is illustrated in Figure \ref{f:lr} in the second forest. There are 3 leaves that begin the tree, corresponding to 3 up-downs that starts the path, and creates a block of size 3 (with weight 4). The rest of the forest must begin with an up-up and belong to $R(n-3)$.
\end{proof}

\begin{figure}
    \centering
    \includegraphics[scale=0.6]{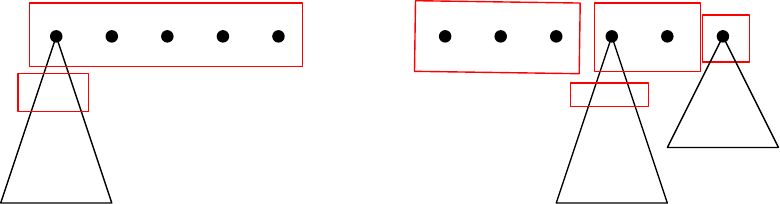}
    \caption{Decomposition of Dyck paths (as viewed as trees) of the forms $L$ and $R$}
    \label{f:lr}
\end{figure}

Now consider the following generating functions. Let $f(x) = \sum_{n \ge 0}\bar{C}(n)x^n$, $g(x) = \sum_{n \ge 0}\bar{L}(n)x^n$ and $h(x) = \sum_{n\ge 0}\bar{R}(n)x^n$. The above recurrences says that these generating functions are linked and that we have a very nice closed form for $f$.

\begin{lemma}
\label{l:gen-func}
    Let $b(x) = (1-x)^{-2}$. Then we have $f = gh +1$, $g = f\cdot(b-1)$ and $h = f/b$. Finally $f = \left(1-\sqrt{1-8x+4x^2}\right)/\left(4x-2x^2\right)$.
\end{lemma}

\begin{proof}
    We first note that $b(x) = \sum_{n\ge 0}(n+1)x^n$. This is easily seen as $b(x)$ is the derivative of $(1-x)^{-1} = \sum_{n\ge 0}x^n$.
    
    Next we note that as $\bar{C}(n) = \sum_{k = 0}^n \bar{L}(k)\bar{R}(n-k)$ for $n\ge 1$, we obtain $f = gh +1$. The constant term accounts for the initial conditions.
    
    Next we expand the recurrence for $\bar{L}$ to obtain:
    \begin{align*}
        g & = \sum_{n=0}^{\infty}\bar{L}(n)x^n = \sum_{n=0}^{\infty}\sum_{l=0}^n\bar{C}(n-l-1)(l+2)x^n = \sum_{l = 0}^{\infty}\sum_{n=l}^\infty\bar{C}(n-l-1)(l+2)x^n\\
        & = \sum_{l = 0}^{\infty}(l+2)x^{l+1}\sum_{n=l}^\infty\bar{C}(n-l-1)x^{n-l-1} = \sum_{l = 1}^{\infty}(l+1)x^lf\\
        & = f\cdot (b-1)
    \end{align*}
    
    Expanding the recurrence for $\bar{R}$ we obtain:
    \begin{align*}
        h & = \sum_{n=0}^{\infty}\bar{R}(n)x^n = \sum_{n=0}^{\infty}\bar{R}(n)x^n - \sum_{n=0}^{\infty}\sum_{l=1}^n\bar{R}(n-l)(l+1)x^n\\
        & = f - \sum_{l=1}^{\infty}\sum_{n=l}^\infty\bar{R}(n-l)(l+1)x^n = f- \sum_{l=1}^{\infty}(l+1)x^{l}\sum_{n=l}^\infty\bar{R}(n-l)x^{n-l}\\
        & = f - (b-1)h
    \end{align*}
    Collect terms and we obtain $f = bh$.
    
    Lastly, we have
    \begin{align*}
        0 & = gh + 1 - f = f^2\frac{(1-x)^{-2}-1}{(1-x)^{-2}} + 1 - f\\
        & = f^2(1-(1-x)^2) - f + 1 = (2x - x^2)f^2-f +1 
    \end{align*}
    Apply the quadratic formula and taking the negative root, we obtain the desired $f = \left(1-\sqrt{1-8x+4x^2}\right)/\left(4x-2x^2\right)$
\end{proof}

We note that the sequence A108524 of OEIS \cite{oeis-sloane} has the same generating function and thus $\bar{C}(n)$ is exactly A108524. Thus we can finally prove our desired lower bound for the number of beer proper interval graphs.

\begin{theorem}
\label{t:lower-bound}
    The number of beer proper interval graphs on $n$ vertices is asymptotically $(4+2\sqrt{3})^n\cdot$poly$(n,1/n)$. 
    Therefore to represent a beer proper interval graph $G$ which is able to support $\qadj$ and $\qbdist$ will require at least $\log( 4+2\sqrt{3})n - o(n) \approx 2.9n$ bits in the worst case.
\end{theorem}

\begin{proof}
    Consider a connected proper interval graph. The weight assigned to it is the weight of the distance tree after we add the dummy root and drop the real root as the first child of the dummy root.
    
    Ignoring the dummy root, this is a particular forest on $n$ nodes and thus is counted in $\bar{C}(n)$. Hence $\bar{C}(n)$ is an upper bound on the number of beer connected proper interval graphs.
    
    Conversely, if we simply delete the root of the distance tree, we obtain a Dyck path on $n-1$ vertices, and thus the weight is counted in $\bar{C}(n-1)$. The root of the tree can only increase the weight (either by being in a block by itself or increasing the size of the first block by 1). Thus $\bar{C}(n-1)$ is a lower bound.
    
    By our generating function, $\bar{C}(n) = (4+2\sqrt{3})^n\cdot$poly$(n,1/n)$, and it does not change asymptotically if we change the number of vertices by 1 (as that only induces a poly$(n,1/n)$ factor). Hence the number of beer connected proper interval graphs on $n$ vertices is asymptotically $(4+2\sqrt{3})^n\cdot$poly$(n,1/n)$.
    
    To convert this to count all proper interval graphs, we note that for any proper interval graph, the blocking scheme is an upper bound on the number of beer vertex patterns (as we are not considering a number of automorphisms in our count - the ones that swap isomorphic components), and thus $(4+2\sqrt{3})^n\cdot$poly$(n,1/n)$ is an upper bound. Clearly, connected proper interval graphs is a subset of all proper interval graphs and $(4+2\sqrt{3})^n\cdot$poly$(n,1/n)$ is a lower bound.
\end{proof}

\bibliography{paper}

\appendix
\section{Bounding the Number of Beer Proper Interval graphs}
\label{s:bounding}
Here we show a method to bound the number of beer proper interval graphs. As we have shown an tight bound in Section \ref{s:lower-bound}, the work here is obsolete, but the ideas contained may be useful in their own right. Furthermore the methods used here seems to be a good starting point in constructing a data structure for beer proper interval graphs using $<3n$ bits.

\subsection{Representation of Proper Interval Graphs}

As Acan et al. \cite{DBLP:journals/algorithmica/AcanCJS21} showed in their data structure, a proper interval graph can be represented using a length $2n$ bitvector, where the $i$-th vertex's left and right end point are the indices of the $i$-th 1 and 0 respectively.

A maximal clique is an clique that is maximal under subsets. An equivalent characterization of Interval graphs is that the set of maximal cliques can be linearly ordered so that for any vertex $v$, the set of maximal cliques containing $v$ is contiguous in the ordering \cite{Fulkerson1965IncidenceMA}. Furthermore the graph is a proper interval graph if this order is unique up to reflection.

To obtain an ordering of the maximal cliques from the bitvector, we consider the 1s and 0s as blocks. The end points of of the $i$-th vertex is now the block number that the $i$-th 1 or 0 belongs to.

We will refer to the intervals generated by the bitvector indices as the bitvector representation of the graph and by the intervals generated from the block representation as the clique representation of the graph.

\begin{lemma}
    Given a proper interval graph, and its interval representation using a length $2n$ bit vector,
    there are $k$ maximal cliques where $k$ is the number of blocks of 1s (and 0s) in the bit vector. If we view the $i$-th vertices' left and right endpoints by the block that the $i$-th 1 and 0 belongs to, then $i$-th maximal clique contains all vertices whose intervals contain $i$.
\end{lemma}
\begin{proof}
To see that this indeed gives us the desired maximal cliques, consider any maximal clique $C$. Then the intervals of the vertices $v \in C$ pairwise intersect and thus there exists some number $l$ that belongs to all  the intervals. Conversely all intervals containing $l$ forms a clique and must be equal to $C$. Thus all maximal cliques are found by taking some number $l$ and taking all intervals containing $l$.

Let $l$ be the index of the last 1 in some block of 1s in the bitvector. We show that $l$ gives a maximal clique. Let $C_l$ be the set of vertices whose intervals contain $l$, and suppose that there exists another number $l' > l$ (other case is symmetrical) whose clique $C_{l'}$ strictly contains $C_l$. Let $v \in C_{l'}\setminus C_l$. The left endpoint of $v$ is to the right of $l$. Since $l$ is the index of the last 1 in a block of 1s, the left endpoint of $v$ must be to the right of the block of 0s immediately following $l$, and thus so much $l'$. But this block of 0s represents the right endpoint of vertices of $C_l$, which then cannot contain $l'$.

By collapsing the 1s and 0s into blocks, the vertices whose intervals contain $i$ (in the block view) are exactly those whose intervals (in the bitvector view) contain $l_i = $  the index of the last 1 in the $i$-th block of 1s.

Conversely, consider any index $i$ such that it is not the last 1 of a block. There are 3 cases, $i$ is between two numbers in the bitvector of the form $11,00, 01$. In the first case and third cases, moving $i$ to the right past the next 1 increases the clique (that 1 represents a new vertex whose interval now contains $i$). In the second case, moving $i$ to the left past the 0 increases the clique (that 0 represents the right endpoint of some vertex that now contains $i$).
\end{proof}

\begin{example}
The proper interval graph represented by the following bit sequence $1101011000$ has maximal cliques
$\{1,2\},\{2,3\},\{3,4,5\}$ and in this arrangement, the cliques that contain any vertex are consecutive.
\end{example}

It is obvious each length $2n$ bitvector that encodes a proper interval graphs must be balanced (at any index, the number of 1 preceding must be at least the number of 0s preceding) so that each interval's right coordinate is larger than its left coordinate, and thus can be viewed as a Dyck path. Consider any index where the Dyck path touches the line $x = 0$. Any vertex represented by an interval to the left of this point does not intersect interval of a vertex to the right of this point, and thus the graph is disconnected. A proper interval graph is connected then if the Dyck path never touches the line $x = 0$, except at the two end points. It is well known that the number of such Dyck paths of length $2n$ is the number of ordinary Dyck paths of length $2(n-1)$ with the bijection - removing the first and last steps.

Let $k$ be the number of maximal cliques (and the number of blocks of 1s and 0s in the bitvector representation). Then the sizes of the blocks of 1s forms a composition of $n$. We will represent this by a set of $k-1$ barriers that split the $n$ nodes into $k$ parts. The positions that do not have barriers will be denoted by the set $R_l$ ($l$ for the left end point of intervals). Similarly, the sizes of the blocks of 0s is also a composition of $n$ and can be represented by a set of $k-1$ barriers. The positions which do not have barriers will be denoted by $R_r$ ($r$ for right end points). Finally, let $R_I = R_l \cap R_r$ be the intersection of the two.

Conversely, given two compositions of $n$, we can recover the block sizes and thus the bitvector.

\begin{lemma}
    Two sets $R_l, R_r$ represents a proper interval graph if at any index $i$, the set $R_l(i) = \{x \in R_l; x < i\}$ is at least as large as the set $R_r(i) = \{x \in R_r; x < i\}$.
\end{lemma}
\begin{proof}
    Consider the clique representation of a proper interval graph. For any vertex $i$, we must have that the left end point is at most equal to the right end point. Translating to $R_l, R_r$, at any index $i$, the number of barriers preceding $i$ in $R_l$ is at most that of the number of barriers preceding $i$ in $R_r$ (the block number of the $i$-th 1/0 is equal to the number of barriers preceding $i$ + 1 in $R_l/R_r$ respectively). Thus the number of positions which do not have barriers in $R_l: R_l(i)$ must be at least $R_r(i)$.
\end{proof}

Given two sets $R_r, R_l$, we say that they satisfy the Dyck path property if they represent a proper interval graph, and we will denote $R_r, R_l$ as the barrier or composition representation of the graph.

\begin{example}
In our graph above whose bitvector representation was $1101011000$, the sizes of the blocks of 1s were $2,1,2$. Thus $R_l$ would be $0010100$ (or as a set $\{1,4\}$), where the $1$ represents the barrier and the 0 represents the size of the composition.

$R_r$ would be $0101000$ (or as a set $\{3,4\}$) as the sizes of the blocks of 0s are $1,1,3$.

Finally $R_I$ would be $01010100$ as only position 4 is shared among $R_l, R_r$.
\end{example}

\subsection{A Lower Bound}
We will derive a subset of proper connected interval graphs that a) there are a lot of graphs in the subset, so that the number of bits to represent them is large and b) the number of bits needed to store the beer vertices is large, which consequently means the number of twin vertices is small.

As our lower bound will be information theoretic, we will be taking the $\log$. The number of graphs will be exponential in $n$ and thus any poly$(n,1/n)$ factors will be lower order terms. For simplicity, we will ignore them. Furthermore, proving a bound on a graph of $n+c$ vertices for constant $c$ will also introduce lower order terms, and thus we will for simplicity ignore any constant increase in graph size.

We will use the barrier representation of a proper interval graph, $R_l, R_r$, as their structure and in particular $R_I$ allows us to compute the sizes of the equivalence classes of twin vertices nicely.

\begin{lemma}
    Two vertices $i$ and $i+1$ are twins if and only if there is no barrier between them in $R_I$.
\end{lemma}
\begin{proof}
    Two vertices $i$, $i+1$ has no barrier between them in $R_I$ $\Leftrightarrow$ $i, i+1$ has no barrier between them in $R_l$ and $R_r$ $\Leftrightarrow$ $i, i+1$'s left endpoints belong to the same block of 1s and their right endpoints belong to the same block 0s $\Leftrightarrow$ $i, i+1$ have the same set of maximal cliques $\Leftrightarrow$ $i, i+1$ are twins.
\end{proof}

Thus the composition defined by $R_I$ is exactly the sizes of the equivalence classes of twin vertices.

\begin{example}
    Again in our graph represented by the bitvector $1101011000$, the only twin vertices are $4,5$ since $R_I = 01010100$.
\end{example}

For a graph on $n$ vertices, there are $n-1$ locations to place barriers. To make the calculations cleaner, we will consider $n+1$ vertices so that there are $n$ locations to place barriers.

\begin{lemma}
    Let $R_r, R_l$ satisfy the Dyck path property for proper interval graphs on $n+1$ vertices. Let $x = |R_I|$ and $y = |R_r \setminus R_I|$. Then the number of such $R_r, R_l$ is
    \[\sum_{x = 0}^n\sum_{y=0}^{(n-x)/2}\binom{n}{x}\binom{n-x}{2y}\binom{2y}{y}\frac{1}{y+1}\]
\end{lemma}
\begin{proof}
    We first select the $x$ positions of $R_I$. Each of $R_l\setminus R_I$, $R_r\setminus R_I$ have size $y$, so from the remaining $n-x$ positions we select $2y$ elements for their union. Finally out of the $2y$ elements, we decide which set each element belongs.
    
    To satisfy the Dyck path property, the $2y$ elements distributed to $R_l, R_r$ must satisfy the Dyck path property. The number of ways to do this is exactly the number of Dyck paths of length $2y$ which is the $y$-th Catalan number $\binom{2y}{y}/(y+1)$.
\end{proof}

Let $S_{x,y} = \{R_l,R_r; |R_I|= x, |R_l| = |R_r| = x + y, R_l,R_r \text{satisfies the Dyck path property}\}$. Then we have shown that $|S_{x,y}| = \binom{n}{x}\binom{n-x}{2y}\binom{2y}{y}\frac{1}{y+1}$.

We also note that as an aside, we have proven the following Catalan number identity ($C_k$ denoting the $k$-th Catalan number), which may be useful for other applications.

\[C_{n+1} = \sum_{x = 0}^n\sum_{y=0}^{(n-x)/2}\binom{n}{x}\binom{n-x}{2y}C_y\]

Let $G$ be any proper interval graph on $n+1$ vertices and let $R_l, R_r$ be the Dyck path representation. Consider the following graph $G'$ on $n+3$ vertices, whose Dyck path representation is $R_l' = 0R_l10$, $R_r' = 01R_r0$. Note that by removing a barrier in $R_l'$ at the first index and not in $R_r'$ we obtain the following property: at any index $i$, $|R_l'(i)|$ is strictly greater than that of $|R_r'(i)|$ and this translates to a Dyck path that never touches the $x$-axis. Therefore, $G'$ is a connected proper interval graph. Furthermore $R_I' = 01R_I10$ so that the sizes of the twin vertex classes are preserved - we add two more of size 1.

Thus if $k_i$ are the sizes of the twin vertex classes of $G$, then
\[|2^n/G'| = 4\Pi_i(k_i+1)\]
We will drop the factor of 4 as it contributes to a lower order term.

Lastly, we wish to investigate the number of beer vertex patterns given only the number of parts in the composition representing the equivalence classes of twin vertices, as that is what $R_I$ gives us.

Let $g(R_l, R_r) = g(R_I)$ be the number of beer vertex patterns. That is if $k_i$ are the sizes of the parts in composition defined by $R_I$, then $g(R_I) = \Pi_i(k_i+1)$.

Let $f(n,x)$ be the average over all compositions of $n$ with $n-x$ parts of the number of beer vertex patterns. That is
\[f(n,x) = \sum_y\sum_{R_l,R_r \in S_{x,y}}g(R_l,R_r)/\sum_y|S_{x,y}|\]

Unfortunately, it is difficult to compute $f(n,x)$ exactly, so the best we can do is bound it.
\begin{lemma}
    $2^n(1/2)^x \le f(n,x) \le 2^n(3/4)^x$. Up to poly$(n,1/n)$ factors.
\end{lemma}

\begin{proof}
    Each $g(R_l,R_r)$ computes the number of beer vertex patterns for a composition of $n$ into $n-x$ parts where $x = |R_I|$ (this might be off by 1 but that only contributes a constant factor). Consider two parts of sizes $k_1 < k_2$. These two parts contributes a factor of $(k_1+1)(k_2+1)$. Now consider two parts of $k_1+1,k_2-1$. The factor is now $(k_1+2)(k_2)$. By doing this we have increased our term by $k_2 - k_1 -1 \ge 0$. Thus if we rearrange our composition by evening out the sizes, we achieve a larger total. Conversely, if we concentrate all of the composition into one term, we obtain the smallest total.
    
    Thus $g(R_l,R_r)$ is maximized when all parts are as equal as possible and $g(R_l,R_r)$ is minimized when all parts have size 1 except the last part which has $n-x+1$. As $f(n,x)$ is the average of all composition, it is bounded by the largest valued compositions and the smallest valued composition.
    
    Thus the lower bound for $f(n,x)$ is the composition $(1,1,\ldots,n-x+1)$ which has value approximately $2^{n-x}(n-x+2)$. Removing poly factors we obtain the desired term $2^n(1/2)^x$.
    
    For the upper bound, consider $0 \le x \le n/2$. In this region the maximum valued composition is $(1,\ldots,2)$ where there are $n-2x$ 1s and $x$ 2s. This composition has a total value of $2^{n-2x}3^{x} = 2^n(3/4)^x$.
    
    For $n/2\le x \le 2n/3$, the maximum valued composition is $(2,2,\ldots,3,3)$. To see that the maximum of $2^n(3/4)^x$ holds, we show that when $x$ increases by 1, the total decreases by a factor of at least $3/4$. To see this in this region, when $x$ increases, we replace 3 parts of 2 by 2 parts of 3. That is we replace a factor of $27 = 3^3$ by $16 = 4^2$. As $16/27 < 3/4$ this holds.
    
    For larger $x$ we replace $k+1$ copies of $k$ by $k$ copies of $k+1$. The values are $(k+1)^{(k+1)}$ and $(k+2)^k$ which in all cases decrease the total by at least $(3/4)$.
\end{proof}

We are finally ready to prove the main result for this section.

\begin{theorem}
    To represent a beer proper interval graph $G$ which is able to support $\qadj$ and $\qbdist$ will require at least $(\log 7)n - o(n) \approx 2.81n$ bits in the worst case.
    
    Furthermore, the lower bound cannot be greater than $n\log 15/2 \approx 2.91n$ bits.
\end{theorem}
\begin{proof}
    First we consider the lower bound. For each beer proper interval graph $G$ on $n+1$ vertices, we apply the transformation to make it connected on $n+3$ vertices. We will also use the approximation $4^y$ to the $y$-th Catalan numbers as that is good enough up to poly$(n,1/n)$ factors. We will also drop the factor 4 that arises in the transformation. Thus the number of beer connected proper interval graphs $N$ on $n+3$ vertices is at least
    
    \begin{align*}
        & \sum_{x=0}^n\sum_{y=0}^{(n-x)/2}\sum_{R_l,R_r\in S_{x,y}} g(R_l,R_r)\\
        &\approx \sum_{x=0}^n\binom{n}{x}f(n,x)\sum_{y=0}^{(n-x)/2}\binom{n}{x}4^y\\
        &\approx \sum_{x=0}^n\binom{n}{x}f(n,x)3^{n-x}\\
        &\ge 6^n \sum_{x=0}^n\binom{n}{x}(1/6)^x\\
        &= 6^n(7/6)^n = 7^n
    \end{align*}
    Taking the log we see that $\log(N) \ge n\log 7 \approx 2.81n$.
    
    On the other hand, for every beer proper interval graph, we may compute $g(R_l,R_r)$ which is exact if it is connected but is only an upper bound if not, thus we obtain the upper bound:
    \begin{align*}
        & \sum_{x=0}^n\sum_{y=0}^{(n-x)/2}\sum_{R_l,R_r\in S_{x,y}} g(R_l,R_r)\\
        &\approx \sum_{x=0}^n\binom{n}{x}f(n,x)\sum_{y=0}^{(n-x)/2}\binom{n}{x}4^y\\
        &\approx \sum_{x=0}^n\binom{n}{x}f(n,x)3^{n-x}\\
        &\le 6^n \sum_{x=0}^n\binom{n}{x}(1/4)^x\\
        &= 6^n(5/4)^n = (15/2)^n
    \end{align*}
    Again taking the log we see that $\log(N) \le \log(15/2)n \approx 2.91n$
\end{proof}

\subsection{Improved Lower bound}
In this section we will improve the lower bound attained in the previous section from $n\log 7$ to $n\log (6+\sqrt{2})$.

We begin with our counting identity:

\[C_{n+1} = \sum_{x = 0}^n\sum_{y=0}^{(n-x)/2}\binom{n}{x}\binom{n-x}{2y}C_y\]

and rewrite it by switching the order of summation:

\[C_{n+1} = \sum_{y=0}^{n/2}\sum_{x=0}^{n-2y}\binom{n}{2y}\binom{n-2y}{x}C_y\]

With the interpretation of first choosing the $2y$ elements of $R_l\cup R_r \setminus R_I$, then choosing which set of $R_l$ and $R_r$ each of these element goes - which again must be a Dyck path on their own. Finally among the remaining elements we choose $x$ of them to be in the intersection.

In this view, fix $R_l' = R_l\setminus R_I$ and $R_r' = R_r \setminus R_I$, with $R_I \subseteq [n] \setminus (R_l'\cup R_r')$.

Define $T = \{(R_l'\cup R_I, R_r'\cup R_I); R_I \subseteq [n] \setminus (R_l'\cup R_r')\}$ be the set of $R_l, R_r$ that we can obtain.

Let $k_1,k_2\ldots,k_{2y+1}$ be the composition defined by $R_I = [n]\setminus S$. Then for any smaller $R_I$, the effect on the partition is to split the parts $k_i$ into smaller parts. Viewing part separately, we see that over all $R_I$, we obtain all partitions of each $k_i$ independently.

Thus $\sum_{(R_l,R_r) \in T} g(R_l,R_r) = \Pi_{i=1}^{2y+1} (h(k_i))$ where $h(k_i)$ denotes the sum over all partitions of $(p_1,\ldots, p_j)$ of $k_i$ elements where the value of each partition is of course $(p_1+1)(p_2+1)\ldots(p_j+1)$.

$h(k)$ follows the recurrence $h(k) = \sum_{i=2}^{k+1}i\cdot h(k-i+1)$ by looking at the size of the last part of the composition. Furthermore it follows the recurrence $h(k) = 4h(k-1) - 2h(k+1)$ and has the close form formula $h(k) = \frac{(2+\sqrt{2})^{k+1}-(2-\sqrt{2})^{k+1}}{4\sqrt{2}}$. $h(k)$ is the sequence A003480 of OEIS \cite{oeis-sloane}.

As $\frac{(2+\sqrt{2})^{k+1}-(2-\sqrt{2})^{k+1}}{4\sqrt{2}} = \frac{(2+\sqrt{2})^{k+1} (1-\frac{1}{3\sqrt{2}})^{k+1}}{4\sqrt{2}} \ge \frac{(2+\sqrt{2})^{k+1} (1-\frac{1}{3\sqrt{2}}^{2})}{4\sqrt{2}} = (2+\sqrt{2})^k(2-\sqrt{2})$ we obtain a nice form for 

\[\sum_{(R_l,R_r) \in T} g(R_l,R_r) = \Pi_{i=1}^{2y+1} (h(k_i)) \ge (2+\sqrt{2})^n(2-\sqrt{2})^{2y+1}\].

\begin{theorem}
    To represent a beer proper interval graph $G$ which is able to support $\qadj$ and $\qbdist$ will require at least $(\log 6+\sqrt{2})n - o(n) \approx 2.89n$ bits in the worst case.
    
    Furthermore, the lower bound cannot be greater than $n\log 15/2 \approx 2.91n$ bits.
\end{theorem}
\begin{proof}
    We have already proven the upper bound in the previous section. For the lower bound, we again consider all Dyck path representation for proper interval graphs on $n+1$ vertices and transform them into connected proper interval graphs on $n+3$ vertices. The number of beer connected proper interval graphs is at least
    
    \begin{align*}
        &\sum_{y=0}^{n/2}\sum_{x=0}^{n-2y}\binom{n}{2y}\binom{n-2y}{x}C_y g(R_l,R_r) \\
        & \ge \sum_{y = 0}^{n/2}\binom{n}{2y}2^{2y}(2+\sqrt{2})^n(2-\sqrt{2})^{2y}\\
        & \ge (2+\sqrt{2})^n (1+2(2-\sqrt{2}))^n\\
        & = (6 + \sqrt{2})^n
    \end{align*}
\end{proof}

\end{document}